\documentclass[12pt, draftclsnofoot, onecolumn]{IEEEtran}
\usepackage{float}
\usepackage{graphicx}
\usepackage{clrscode}
\usepackage{stfloats}
\usepackage{paralist}
\usepackage{tikz}
\usepackage{color}

\usetikzlibrary{arrows,automata}
\usepackage{cite}
\usepackage{comment}
\usepackage{graphicx}
\usepackage{amsmath, amssymb, mathrsfs, amsfonts, amsthm}
\usepackage{mathtools}
\usepackage{epsfig, epstopdf}
\usepackage[tight,footnotesize]{subfigure}
\usepackage{algorithm}
\usepackage{enumitem}
\usepackage{algorithmic}
\usepackage{hyperref}
\allowdisplaybreaks

\newtheorem{theorem}{Theorem}%[section]

\newtheorem{corollary}{Corollary}

\newtheorem{remark}{Remark}

\def\msP{ \mathsf{P}}

\newcommand{\HTemp}{\mathcal{H}}
\newcommand{\HH}{\mathcal{H}_0}
\newcommand{\HHH}{\mathcal{H}_1}

\newcommand{\Pf}{ \text{P}_{\text{fa}} }
\newcommand{\Pd}{ \text{P}_{\text{d}} }
\newcommand{\Pm}{ \text{P}_{\text{m}} }

\newcommand{\Pfl}{ \text{P}_{\text{fa,local}} }
\newcommand{\Pdl}{ \text{P}_{\text{d,local}} }

\newcommand{\g}{ \text{g} }
\newcommand{\GG}{ \text{G} }
\newcommand{\G}{ \text{G} }
\newcommand{\MM}{ \mathcal{X}'}
\newcommand{\M}{ \mathcal{X}}

\newcommand{\E}{\mathbb{E}}
\newcommand{\Po}{\mathsf{Poisson}}
\newcommand{\Ex}{\mathsf{exp}}
\newcommand{\LLRset}{\mathcal{S}}

\newcommand{\argmax}{\arg\!\max}

\begin{document}

\title{Cooperative Abnormality Detection via Diffusive Molecular Communications}
\author{Reza Mosayebi$^\dagger$, Vahid Jamali$^{\ddagger}$, Nafiseh Ghoroghchian$^\dagger$, Robert Schober$^{\ddagger}$, Masoumeh Nasiri-Kenari$^\dagger$, and Mahdieh Mehrabi$^\dagger$ \\
   $^\dagger$Sharif University of Technology, Tehran, Iran\\
     $^{\ddagger}$University of Erlangen-Nuremberg, Erlangen, Germany}

\maketitle

\vspace{-1.5cm}

\begin{abstract}
In this paper, we consider abnormality detection via diffusive molecular communications (MCs) for a network consisting of several sensors and a fusion center (FC). If a sensor detects an abnormality, it injects into the medium a number of molecules which is proportional to the sensed value. Two transmission schemes for releasing molecules into the medium are considered. In the first scheme, referred to as DTM, each sensor releases a different type of molecule, whereas in the second scheme, referred to as STM, all sensors release the same type of molecule. The  molecules released by the sensors  propagate through the MC channel and some may reach the FC where the final decision regarding whether or not an abnormality has occurred is made.  We derive the optimal decision rules for both DTM and STM. However, the optimal detectors entail high computational complexity as log-likelihood ratios (LLRs) have to be computed. To overcome this issue, we show that the optimal decision rule for STM can be transformed  into an equivalent low-complexity decision rule. Since a similar transformation is not possible for DTM, we propose simple low-complexity sub-optimal detectors based on different approximations of the LLR. The proposed low-complexity detectors are more suitable for practical MC systems than the original complex optimal decision rule, particularly when the FC is a nano-machine with limited computational capabilities. Furthermore, we analyze the performance of the proposed detectors in terms of their false alarm and missed detection probabilities. Simulation results verify our analytical derivations and reveal interesting insights regarding the trade-off between complexity and performance of the proposed detectors and the considered DTM and STM~schemes.
\end{abstract}


\section{Introduction}

The progress in the design of nano-scale machines over the past decade has motivated researchers to study the concept of nano-communications. Inspired by biological systems, diffusion-based molecular communication (MC) systems  have been proposed as a potential solution for communication in nano-networks where molecules are used as information carriers  \cite{MC_Book,CellBio}. Nano-networks are envisioned to facilitate revolutionary applications in areas such as biological engineering, healthcare, and environmental monitoring~\cite{Akyl_MCNet}.

\subsection{Motivation}
In recent years, there has been a significant amount of work on various aspects of MC systems, including transmitter and receiver design  \cite{Modulation_Akyl,Reza_Receiver,Vahid_NanoCOM,Arman_Receiver}, multiple access protocols \cite{MacBrRelay}, and network layer issues \cite{NetLayer}. However, the problem of \textit{abnormality detection}, which is one of the key challenges in many of the applications envisioned for nano-networks such as environmental and health monitoring and disease diagnosis, has not been fully investigated yet. For example, to enable smart drug delivery, first an abnormality has to be detected and its progress has to be monitored \cite{Drug}. Then, depending on the condition of the target site, a drug can be released at an appropriate rate. This motivates the investigation of abnormality detection at micro-scale using MCs.

\subsection{Prior Work}

Abnormality detection has been extensively studied in different fields, see e.g. \cite{Poor_Book_Detection,Survey_Anomaly}. In this context, abnormalities are also referred to as anomalies, outliers, exceptions, aberrations, surprises, and peculiarities in various areas of application including failure detection in computer science \cite{Computer_Failor}, fraud detection for credit cards \cite{Fraud_Detection}, and the segmentation of signals in biomedical applications \cite{Bio_Detection}. In most applications, abnormality detection is crucial and requires a high degree of accuracy. A widely-adopted strategy to increase the detection accuracy is through cooperative sensing where several distributed sensors send their sensing results to a common fusion center (FC) which makes the final decision \cite{Survey_Sensing,Sensing_Vahid}. 

Abnormality detection via MC introduces certain new challenges which are not encountered in other fields. In particular, for collaborative abnormality detection in MC applications, the communication between the sensors and the FC is challenging due the inherent randomness of the MC channel. Very few prior works have considered this problem.
Recently in \cite{Adam_Sensing_MC}, a cooperative MC system has been considered, where one transmitter and several receivers send their local \emph{hard decisions} about a transmitted bit to an FC which makes the final decision. However, \cite{Adam_Sensing_MC} does not consider the abnormality detection problem in particular. In another recent work \cite{Lahouti_Detection}, the abnormality detection problem in MC sensor networks is studied and the \emph{sub-optimal OR} fusion rule  is employed to combine the observations received at the FC based on \textit{hard decisions} made at the sensors. In \cite{Lahouti_Detection}, it is assumed that all sensors employ the same type of molecules and the reporting channels, i.e., the sensor-FC channels, are modeled as \textit{additive white Gaussian noise} (AWGN) channels. In the next subsection, we discuss in detail how this paper expands  \cite{Lahouti_Detection} in several important directions, e.g., by allowing soft decisions at the sensors, making more realistic assumptions regarding the MC channel, deriving optimal detectors, and developing low-complexity sub-optimal detectors.

\subsection{Contributions}

We consider collaborative abnormality detection where multiple sensors sense a surface, e.g. an area of tissue, and send their \emph{soft} noisy  sensing values to an FC via diffusion-based MC. We consider two different transmission schemes where the sensors employ different types of molecules (DTM) and the same type of molecule (STM), respectively. The considered DTM and STM schemes provide a trade-off between complexity and performance. Moreover, we model the diffusion channel as a Poisson channel with an arbitrary memory length, i.e., inter-symbol interference (ISI) is present and we assume that the environmental background noise is also Poisson distributed as in \cite{Reza_Receiver,Yilmaz_Poiss,Hamid_letter}. We note that the Poisson model is more realistic for MC systems with molecule counting receivers than the AWGN model assumed in \cite{Lahouti_Detection}, see \cite{Hamid_letter}. 

For the considered MC system, we first derive optimal fusion rules for both DTM and STM. We note that the optimal detectors entail high computational complexity as log-likelihood ratios (LLRs) have to be computed. To overcome this issue, we show that for STM, the optimal decision rule can be transformed  into an equivalent low-complexity decision rule. Since a similar transformation is not possible for DTM, we propose several simple low-complexity sub-optimal detectors based on different approximations of the LLR. The proposed low-complexity detectors are more suitable for practical MC systems than the complex optimal decision rule, particularly when the FC is a  nano-machine with limited computational capabilities. Furthermore, we propose an analytical approach for numerical evaluation of the performance of the proposed optimal and sub-optimal detectors in terms of their false alarm and missed detection probabilities. In addition, the performance of some of the detectors is analyzed in closed form. We further provide asymptotic performance bounds for the presented detectors and derive approximate error exponents for a large number of sensors. These bounds allow us to compare the asymptotic performance of different detectors. Simulation results verify our analytical derivations and provide interesting insights regarding the trade-off between complexity and performance for the proposed optimal and sub-optimal detectors and for the considered DTM and STM schemes.

\subsection{Organization}
The remainder of this paper is organized as follows. In Section II, we describe the system model, the sensing model, and the reporting channel model. In Sections III, we derive optimal detectors for both the DTM and STM reporting schemes, whereas, in Section IV, several sub-optimal low-complexity detectors for DTM are proposed. The performance analysis  and asymptotic performance bounds for the proposed schemes are presented in Section V. Finally, Section VI provides extensive numerical performance results and comparisons, and Section VII concludes the paper.

%%%%%%%%%%
%%%%%%%%%%
%%%%%%%%%%
%%%%%%%%%%
\section{System Model}\label{sysmod}
 We consider a hypothesis testing problem for a system consisting of $M$ identical sensors, each monitoring a part of a target, e.g., an area of tissue, and an FC, see Fig. \ref{fig.general_SN_model}. Let $\HH$ and $\HHH$ denote the normal and abnormal hypotheses, respectively. The goal is to decide at the FC whether the normal or the abnormal hypothesis is true based on the observations received from the sensors. We describe the sensing model, the reporting channel model, and  the FC in Sections \ref{IA}, \ref{IB}, and \ref{IC}, respectively.

\subsection{Sensing Model}\label{IA}

We adopt a general and abstract model for the sensing process. In particular, let $X^{\{m\}}\in[0,1]$ be a random variable (RV) which models the sensed value at the $m$-th sensor and let $x_m$ be a realization of RV $X^{\{m\}}$. Hereby, small and large values of $x_m$ are used to indicate that the sensing observation at sensor $m$ leans towards hypotheses $\HH$ and $\HHH$, respectively. In fact, for ideal (noise-free) sensors, we have $x_m=0$ and $x_m=1$ for hypotheses $\HH$ and $\HHH$, respectively; however,  for non-ideal (noisy) sensors, we have in general $x_m\in[0,1]$ for both hypotheses. Due to practical considerations, we assume $x_m\in\mathcal{X}= \{0,1/(L-1), \dots,(L-2)/(L-1), 1\}$ where $\mathcal{X}$ is the set of $L$ possible sensed values at the sensors. Given each hypothesis, RVs $X^{\{m\}},\,\,\forall m$, are assumed to be independent and identically distributed (i.i.d.), i.e., the sensors are assumed to be spatially uncorrelated. In particular, we mathematically model $X^{\{m\}}$ as
\begin{align} \label{z} 
X^{\{m\}}\sim
\begin{cases}
\g_0(x_m),\,\, &\text{under hypothesis $\HH$} \\
\g_1(x_m),\,\, &\text{under hypothesis $\HHH$}
\end{cases}
\end{align}
where $\g_0(\cdot)$ and $\g_1(\cdot)$ denote the probability distribution functions (PDFs) of $x_m,\,\,\forall m$,  under hypotheses $\HH$ and $\HHH$, respectively. For ideal sensors, we have $\g_0(x_m)=\delta(x_m)$ and $\g_1(x_m)=\delta(x_m-1)$ where $\delta(\cdot)$ denotes the Dirac delta function. However, for practical sensors, we expect that the PDFs $\g_0(x_m)$ and $\g_1(x_m)$ are decreasing and increasing functions of $x_m$, respectively. As a special case, for $L=2$,  which corresponds to a hard decision at the sensor, the probabilities of false alarm and missed detection at each sensor are denoted by $p_0$ and $p_1$, respectively. That is $\g_0(0)=1-p_0$, $\g_0(1)=p_0$ and $\g_1(0)=p_1$, $\g_1(1)=1-p_1$. 

\subsection{Reporting Model}\label{IB}

We consider a time-slotted transmission with time slots of duration $T$, and a reporting period of $[0, NT]$ where $N$ is the number of time slots. We assume that, after sensing, each sensor releases $x_mA^{\max}$ molecules into the environment at the beginning of each of the $N$ time slots within the reporting period\footnote{Practical sensors may not be able to release a large number of molecules at once due to a limited molecule production rate. To cope with this issue, instead of releasing a large number of molecules in one time slot, we assume that each sensor releases a small number of molecules in \textit{multiple} time slots.}. Here, $A^{\max}$ denotes the maximum number of molecules that the sensors are able to release into the channel. We consider two different reporting schemes depending on the type of molecule released by the sensors, namely DTM and STM, see Fig.~\ref{fig.general_SN_model}. In DTM, each sensor releases a different type of molecule, whereas in STM, all sensors release the same type of molecule. 

%The DTM scheme generally provides a better performance than the STM scheme at the expense of a higher complexity due to the required different types of molecules, see Section~VI. 

\begin{figure}[t]
\centering 
\includegraphics[scale=0.8]{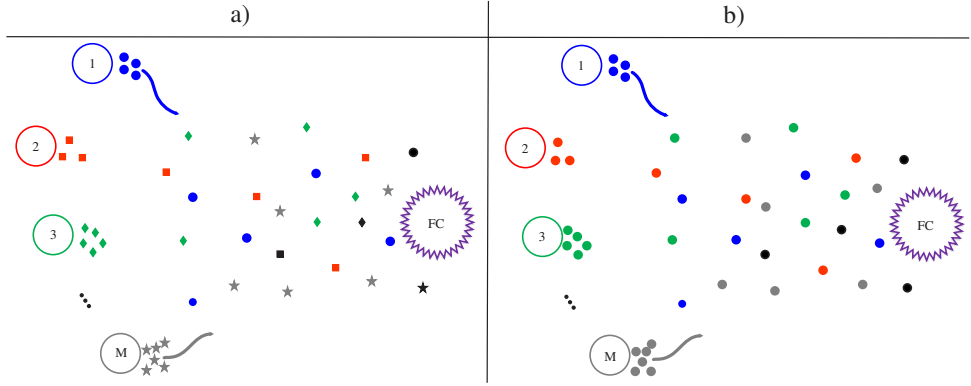} 
\vspace{-0.5cm}
\caption{The considered sensor network consisting of $M$ sensors and an FC for a) the DTM scheme and b) the STM scheme. Shape and color of the molecules represent the type of molecule and the sensor which releases the molecule, respectively. The environmental noise molecules are represented by black circles. }
\label{fig.general_SN_model}
\vspace{-0.3cm}
\end{figure}

The molecules released by the sensors propagate through the channel and are observed at the FC. We assume that a molecule is absorbed at the FC when it hits the surface of the FC. For simplicity of analysis, we assume that the channels between the sensors and the FC are statistically identical, i.e., the distances between all sensors and the FC are identical. In particular, let $h_k$ denote the hitting probability, i.e., the probability that the molecule released in the current time slot hits the receiver during the $k$-th next time slot. The hitting probability $h_k$ depends on the communication medium \cite{Yilmaz_Absorbing}. For instance, for a general three-dimensional environment with a point-source transmitter located at the origin and having a distance of $r_1$ from the center of a spherical receiver with radius $r_2$, the hitting probability is given by \cite{Yilmaz_Absorbing}
\begin{align}
h_k  = \begin{cases}
\frac{r_2}{r_1} \text{erfc} \left(  \frac{r_1-r_2}{\sqrt{4DT}} \right), \quad \text{if}\,\,k=0 \\
\frac{r_2}{r_1} \left(  \text{erfc} \left( \frac{r_1-r_2}{\sqrt{4D(k+1)T}} \right)-\text{erfc} \left( \frac{r_1-r_2}{\sqrt{4DkT}} \right)    \right), \quad \text{if}~ k\geq 1,
\end{cases}
\end{align}
where $D$ is the diffusion coefficient of the transmitted molecule and $\text{erfc}(x)= \frac{2}{\sqrt{\pi}} \int_{x}^{\infty} \text{e}^{ -t^2}\text{d}t$ is the complementary error function. For general MC channels, the hitting probabilities $h_k,\,\,\forall k$, can be estimated  using  training sequence-based channel estimators \cite{Vahid_CSI}.

Let $Y_n^{\{m\}}$ be a RV which models the number of molecules of the type that the $m$-th sensor employs observed at the FC in the $n$-th time slot and let $y_n^{\{m\}}$ be a realization of RV $Y_n^{\{m\}}$. Therefore, the $y_n^{\{m\}}$, $m \in \{1, \ldots, M \}$, constitute the received signals at the FC for DTM. In contrast, for STM, the FC is not able to distinguish between the molecules released by different sensors as all sensors employ the same type of molecule. In this case, the total number of molecules absorbed at the FC in the $n$-th time slot is modelled by RV $Y_n=\sum_{m=1}^M Y_{n}^{\{m\}}$. Hereby, we denote a particular realization of RV $Y_n$ by $y_n$ which constitutes the received signal at the FC. Assuming synchronous transmission \cite{Vahid_Sync}, the input-output relationships of the channels between the sensors and the FC are modeled as independent Poisson channels with additive Poisson noise \cite{Reza_Receiver,Yilmaz_Poiss,Reza_ITA}, i.e., 
\begin{align}\label{Eq:Yn}
\begin{cases}
Y_n^{\{m\}} \sim \Po\left(J_m + \sum_{k=0}^{n}h_k x_m A^{\max}\right), \quad m \in \{ 1, \cdots, M\}, &\text{for DTM} \\
Y_n \sim \Po\left(J + \sum_{m=1}^M\sum_{k=0}^{n}h_k x_m A^{\max}\right), \quad &\text{for STM}
\end{cases}
\end{align}
where $J_m$ and $J$ are the expected numbers of environmental  noise molecules observed at the FC in each time slot. We assume that the expected number of noise molecules is identical for all types of signaling molecules employed, i.e., $J_m=J,\,\,\forall m$, holds. From (\ref{Eq:Yn}), we observe that both $Y_n^{\{m\}}$ and $Y_n$ comprise  three independent terms: \textit{i)} $J$, the noise term, \textit{ii)} $h_0 x_m A^{\max}$, the expected number of received molecules due to the current  transmission, and \textit{iii)} $\sum_{k=1}^{n}h_k x_mA^{\max}$, the expected number of received molecules due to previous transmissions. 

The term $ A_n \triangleq \sum_{k=0}^{n}h_n A^{\max}$ increases with increasing $n$ and saturates to constant $A \triangleq \sum_{k=0}^{\infty}h_k A^{\max}$ for $n \rightarrow \infty$. The speed of convergence depends on the environment, i.e., the values of the $h_k,\,\,\forall k$. If the channel has a finite memory of $\bar{k}T$, i.e., $h_{k}=0,\,\,\forall k\geq\bar{k}$, the transient effect vanishes after $\bar{k}$ time slots, and then, the expected number of molecules received in each time slot will be constant, i.e., $A_n = A,\,\, \forall n > \bar{k}$. The transient period during $0 \leq n \leq \bar{k}$, adds a non-uniformity to the hypothesis testing problem. However, for a channel memory $\bar{k}T$ much less than the reporting period of $NT$, the transient period can be neglected. In fact, there are several options for mitigating the effect of the transient period: \textit{i)} choosing $N\gg\bar{k}$ so that we can neglect the first $\bar{k}$ time slots without compromising the performance, \textit{ii)} choosing $T$ sufficiently large such that $h_k \approx 0,~ \forall k\geq 1$, and \textit{iii)} employing enzymes \cite{Adam_Enzyme}  to speed up the decaying of the channel impulse response (CIR)  as a function of time, see \cite[Section 2]{Vahid_NanoCOM} for further justification.
Using either of these approaches,   the \textit{expected} number of  molecules received at the FC from a given sensor will be constant in each time slot. 
Thus, assuming that observations at the FC in different time slots are identically distributed, for DTM, the probability distribution  of $(Y^{\{m\}}_1, \dots, Y^{\{m\}}_N)$  is given by
\begin{align}\label{Eq:IndSamp}
\msP \left( {Y}^{\{m\}}_{1}= {y}^{\{m\}}_{1}, \dots, {Y}^{\{m\}}_{N}={y}^{\{m\}}_{N}  | \HTemp_i \right)
= \prod_{n=1}^{N} \msP \left( {Y}^{\{m\}}_{n}={y}^{\{m\}}_{n}| \HTemp_i \right),\qquad i\in\{0,1\},
\end{align}
where $\msP(\cdot)$ denotes probability. Eq. (\ref{Eq:IndSamp}) also holds  for STM  after removing superscript $\{m\}$.

\subsection{Fusion Center}\label{IC}

The FC is assumed to be a node with sufficient processing capability that is able to make decisions based on the molecules received from the sensors. A promising approach for implementing such a node is using engineered bacteria \cite{Engineered_Bacteria} where the bacteria can be designed to individually detect and count different types of molecules. Such engineered bacteria may also be able to perform simple processing tasks. To perform more sophisticated processing, one approach is to employ a synthetic/electronic nano-machine with sufficient processing capability to aid the engineered bacteria.  Alternatively, for applications where a nano-machine only collects observations and forwards them to an external processing unit which serves as the FC outside the MC environment, high
computational complexity  may be affordable. This case may apply e.g. in health monitoring when
a computer outside the body may be available for offline processing~\cite{Vahid_CSI}.

\section{Optimal Detector Design}\label{OptimalDetector} 

In this section, we derive the optimum detection rules for both DTM and STM. Let $\vec{Y}$ denote an RV vector modelling the observation vector which contains the numbers of received molecules during all time slots, i.e.,  $Y_n^{\{m\}},\,\,\forall n,m$ for DTM and $Y_n,\,\,\forall n$ for STM, and let $\vec{y}$ be a realization of RV $\vec{Y}$. To derive the optimum decision rule, we must compare the LLR with a threshold denoted by $\gamma$ \cite{Poor_Book_Detection}. Correspondingly, the decision of the optimum detector can be characterized~as
\begin{align} \label{mMAP_d4}
\text{d}_{\text{opt}}=\left\{ 
\begin{array}{l l}
0, & \,\, \text{if $~ \overline{\mathsf{LLR}}(\vec{y}) \leqslant \gamma$} \\
1, & \,\, \text{otherwise} \\
\end{array} \right.
\end{align} 
where $\text{d}_{\text{opt}}=i$ means that the detector selects hypothesis $\HTemp_i,\,\,i=0,1$. In (\ref{mMAP_d4}), $\overline{\mathsf{LLR}}(\vec{y})$ is given by
\begin{align} \label{LLR}
\overline{\mathsf{LLR}}(\vec{y}) = \mathsf{log}\left(\frac{\msP \left( \vec{Y}=\vec{y} | \HHH\right)}{\msP \left( \vec{Y}=\vec{y} | \HH \right)}\right).
\end{align}
In the following, we derive the LLR for the DTM and STM reporting schemes, respectively.

\subsection{DTM Reporting}\label{DO}
 To calculate the LLR in (\ref{LLR}), we have to first derive the conditional distributions of the received molecules at the FC, i.e., $\msP \left( \vec{Y}=\vec{y} | \mathcal{H}_i\right),\,\,i=0,1$. In particular, we have
\begin{align}
\msP \left(\vec{Y}=\vec{y} | \mathcal{H}_i \right) = \prod_{m=1}^{M} \msP \left( \vec{Y}^{\{m\}}=\vec{y}^{\{m\}} | \mathcal{H}_i \right), \quad i=0,1,
\end{align}
where $\vec{Y}^{\{m\}} = [ Y_{1}^{\{m\}}, \dots, Y_{N}^{\{m\}} ]^T$ denotes an RV modelling the DTM observation originating from sensor $m$ at the FC, i.e., the numbers of the molecules received from the $m$-th sensor in the $N$ time slots of the reporting period, and  $\vec{y}^{\{m\}}= [ y_{1}^{\{m\}}, \dots, y_{N}^{\{m\}} ]^T$ denotes a particular realization of $\vec{Y}^{\{m\}}$.
Moreover, for $i=0,1$ and $m=1,\ldots, M$, $\msP \left( \vec{Y}^{\{m\}} | \mathcal{H}_i \right)$ is obtained as
\begin{align}\label{Hi-new}
\msP \left( \vec{Y}^{\{m\}} | \mathcal{H}_i \right) &= \sum_{x_m \in \mathcal{X}} \msP \left( \vec{Y}^{\{m\}}=\vec{y}^{\{m\}} | X^{\{m\}}=x_m,  \mathcal{H}_i \right)
\msP \left(X^{\{m\}}=x_m | \mathcal{H}_i \right) \nonumber \\
&=\sum_{x_m \in \mathcal{X}} \g_i(x_m) \msP \left( \vec{Y}^{\{m\}}=\vec{y}^{\{m\}} | X^{\{m\}}=x_m, \mathcal{H}_i \right) \nonumber \\
&\overset{(a)}{=}\sum_{x_m \in \mathcal{X}} \g_i(x_m) \prod_{n=1}^{N} \msP \left(Y^{\{m\}}_n=y^{\{m\}}_n | X^{\{m\}}=x_m, \mathcal{H}_i \right) 
\nonumber \\
&\overset{(b)}{=}\left( \sum_{x_m \in \mathcal{X}} \g_i(x_m) \Ex \left(-N(x_mA+J)\right) \left(x_mA+J\right)^{\sum_{n=1}^{N}y^{\{m\}}_n} \right) /  \prod_{n=1}^{N}y^{\{m\}}_n!, \,\,
\end{align}
where equality $(a)$ follows from the fact that the observations in different time slots, i.e., $Y_{n}^{\{m\}},\,\,\forall n$, are independent, and 
equality $(b)$ is obtained by substituting $\msP \big( \vec{Y}^{\{m\}}=\vec{y}^{\{m\}} | X^{\{m\}}=\\ x_m, \mathcal{H}_i \big)$ by the Poisson distribution with mean $x_mA+J$. Since the sensors are assumed to be independent, the total LLR is the sum of the LLRs of all sensors \cite{Poor_Book_Detection}, i.e., we have
\begin{align}\label{TotLLR}
\overline{\mathsf{LLR}}(\vec{y}) = \sum_{m=1}^{M} \mathsf{log}\left(\frac{\msP \left( \vec{Y}^{\{m\}} | \HHH\right)}{\msP \left( \vec{Y}^{\{m\}} | \HH \right)}\right) 
\triangleq \sum_{m=1}^{M} \mathsf{LLR}(\vec{y}^{\{m\}}),
\end{align}
where 
\begin{align}\label{llr_loc}
\mathsf{LLR}(\vec{y}^{\{m\}}) = \mathsf{log} \left( \frac{ \sum_{x_m \in \mathcal{X}} \g_1(x_m) \Ex(-N(x_mA+J)) (x_mA+J)^{\sum_{n=1}^{N}y^{\{m\}}_n} }
{ \sum_{x_m \in \mathcal{X}} \g_0(x_m) \Ex(-N(x_mA+J)) (x_mA+J)^{\sum_{n=1}^{N}y^{\{m\}}_n} } \right).
\end{align}
For the special case of $L=2$, i.e., hard decisions at the sensors, the total LLR can be simplified to
\begin{align}\label{TotLLRH}
\overline{\mathsf{LLR}}(\vec{y}) = \sum_{m=1}^{M} \mathsf{log} \left( \frac{ (1-p_1) \Ex(-N(A+J)) (A+J)^{\sum_{n=1}^{N}y^{\{m\}}_n} + p_1 \Ex(-NJ) J^{\sum_{n=1}^{N}y^{\{m\}}_n}}
{p_0 \Ex(-N(A+J)) (A+J)^{\sum_{n=1}^{N}y^{\{m\}}_n} + (1-p_0) \Ex(-NJ) J^{\sum_{i=n}^{N}y^{\{m\}}_n}} \right) .
\end{align} 
%
%***********************************
Using (\ref{TotLLR}), the optimal decision can be obtained from (\ref{mMAP_d4}). 

The following theorem reveals  the monotonicity of  $\mathsf{LLR}(\vec{y}^{\{m\}})$ with respect to $\sigma_{y}^{\{m\}}=\sum_{n=1}^{N}y^{\{m\}}_n$ under some mild conditions. In Section~IV, we exploit this property to present a low-complexity two-stage detector for DTM. 

\begin{theorem}\label{Th2}
For each sensor, $\mathsf{LLR}(\vec{y}^{\{m\}})$ is a monotonically increasing function with respect to the sum of samples, $\sigma_{y}^{\{m\}}=\sum_{n=1}^{N}y_n^{\{m\}}$, under the following mild conditions
\begin{align}
\frac{\g_1\left(x_m\right)}{\g_1\left(x_m'\right)}\geq \frac{\g_0\left(x_m\right)}{\g_0\left(x_m'\right)}, \quad \forall x_m, x_m'\in{\M}, x_m'<x_m.
\end{align}
\end{theorem}
\begin{proof}
The proof is given in Appendix \ref{Th2:Proof}. %A.
\end{proof}
%**
\begin{remark}
For the special case of $L=2$, i.e., hard decisions at the sensors, the only values to be checked are $x_m=1$ and $x_m'=0$. Hence, the condition in Theorem~\ref{Th2} simplifies to
\begin{align}
\frac{1-p_1}{p_1} \geq \frac{p_0}{1-p_0},
\end{align}
resulting in the condition $p_0+p_1\leq1$, which always holds.
\end{remark}
Assuming the probability functions $\g_0(x_m)$ and $\g_1(x_m)$ are decreasing and increasing functions of $x_m$, respectively, the condition in Theorem \ref{Th2} always holds. In fact, it is expected that $\g_0(x)$ and $\g_1(x)$ have the aforementioned properties for practical sensors. For instance, under hypothesis $\HHH$, a sensor with a reasonable performance would be expected to yield $\g_1(x)>\g_1(x')$ for $x>x'$.

\subsection{STM Reporting}

In this subsection, we focus on sensors employing STM reporting, i.e., all sensors employ the same type of molecule for signaling their observations to the FC. Therefore, the FC cannot distinguish between the molecules released by different sensors. In this case, the conditional probability distribution of the received molecules at the FC can be written as
\begin{align}\label{HiSUM}
\msP\left(\vec{Y}=\vec{y} | \mathcal{H}_i\right) 
&= \sum_{x\in \mathcal{X}'}\msP \left(\vec{Y}=\vec{y} \big|X= x, \mathcal{H}_i \right) \msP\left(X= x | \HTemp_{i}\right), \quad i=0,1,
\end{align}
where $X=\sum_{m=1}^{M}X^{\{m\}}$ is the sum of all sensing values. The possible values, $x$, of $X$, lie in the set $\mathcal{X}' \triangleq \{l/(L-1): l=0, 1, \cdots, M(L-1) \}$, i.e., $x\in \mathcal{X}'$. The probability $\msP \left(\vec{Y}=\vec{y} \big|X= x, \mathcal{H}_i \right)$ is obtained as 
\begin{align}\label{ConProb}
 & \msP \left(\vec{Y}=\vec{y} \big|X= x, \mathcal{H}_i \right) 
 \overset{(a)}{=}  \prod_{n=1}^{N} \msP \left(Y_n=y_n \big|X=x,\mathcal{H}_i \right)  \nonumber \\
 &  \overset{(b)}{=}  \prod_{n=1}^{N} \frac{\Ex(-(xA+J))(xA+J)^{y_n}} {y_n !} 
   =\frac{\Ex(-N(xA+J))(xA+J)^{\sum_{n=1}^{N}y_n}}{\prod_{n=1}^{N}y_n!},  
 \end{align}
 where for equality $(a)$, we exploit the mutual independence of the observations, $Y_n,\,\,\forall n$, and for equality $(b)$, the expression for the Poisson distribution with mean $xA+J$ is substituted for $\msP \big(Y_n=y_n \big|X=x,\mathcal{H}_i \big)$. The probability $\msP\left(X= x | \HTemp_{i}\right)$ can be calculated as
 \begin{align}\label{myconv}
\msP\left(X= x | \HTemp_{i}\right) \triangleq \GG_i(x) \sim \left( \overbrace{\g_i(x_m) \otimes \g_i(x_m) \otimes \cdots \otimes \g_i(x_m)}^{M-1 ~\text{times}}  \right),\quad x_m \in \mathcal{X}, i=0, 1,
\end{align}
where ``$\otimes$" represents the convolution operator. 
Using (\ref{ConProb}) and (\ref{myconv}), we can rewrite (\ref{HiSUM}) as
\begin{align}\label{Hi_STM}
\msP(\vec{Y}=\vec{y} | \mathcal{H}_i) 
= \frac{1}{{\prod_{n=1}^{N}y_n!}} \sum_{x \in \mathcal{X}'} \GG_i(l) {\Ex(-N(xA+J))(xA+J)^{\sum_{n=1}^{N}y_n}}\quad i=0, 1.
\end{align}
Therefore, the LLR for STM can be calculated as
\begin{align}\label{LLR2}
\overline{\mathsf{LLR}}(\vec{y}) = \mathsf{log} \left(  \frac{\sum_{x \in \mathcal{X}'}\GG_1(x) ~{\Ex(-N(xA+J))(xA+J)^{\sum_{n=1}^{N}y_n}}}
{\sum_{x \in \mathcal{X}'} \GG_0(x) ~{\Ex(-N(xA+J))(xA+J)^{\sum_{n=1}^{N}y_n}}}
 \right).
\end{align}
Having the LLR in  (\ref{LLR2}), the optimal decision can be readily obtained from (\ref{mMAP_d4}). In the following corollary, we prove the monotonicity of the above LLR function in $\sum_{n=1}^{N}y_n^{\{m\}}$. Then, based on this property, we propose a low-complexity optimal detector.

\begin{corollary}\label{CorolSTM}
The LLR for STM, $\overline{\mathsf{LLR}}(\vec{y})$, given in (\ref{LLR2}) is a monotonically increasing function with respect to the sum of samples, $\sigma_{y}=\sum_{n=1}^{N}y_n$, under the following condition
\begin{align}
\frac{\G_1\left(x\right)}{\G_1\left(x'\right)}\geq \frac{\G_0\left(x\right)}{\G_0\left(x'\right)}, \quad \forall x, x'\in{\MM}, x'<x.
\end{align}
\end{corollary}
\begin{proof}
The proof is similar to that for Theorem~\ref{Th2}. Therefore, in order to avoid repetition, we omit the detailed proof here.
\end{proof}

Based on Corollary~\ref{CorolSTM}, $\overline{\mathsf{LLR}}(\vec{y})$ is strictly increasing in terms of $\sum_{n=1}^{N}y_n$. Let us denote this mapping by $\overline{\mathsf{LLR}}(\vec{y}) = f(\sum_{n=1}^{N}y_n)$ and its inverse by $f^{-1}(\overline{\mathsf{LLR}}(\vec{y}))=\sum_{n=1}^{N}y_n$. Therefore,  comparing $\overline{\mathsf{LLR}}(\vec{y})$ with threshold $\gamma$, cf. (\ref{mMAP_d4}), is equivalent to comparing $\sum_{n=1}^{N}y_n$ with a new threshold $\gamma_{\text{STM}}=f^{-1}(\gamma)$. This leads to the following equivalent optimal detector
\begin{align} \label{OptSTMsimple}
\text{d}_{\text{opt}}=
\begin{cases} 
0, & \,\, \text{if} \,\, \sum_{n=1}^{N}y_n \leqslant \gamma_{\text{STM}} \\
1, & \,\, \text{otherwise.} \\
\end{cases}
\end{align} 
We note that the above detector is much simpler than that in (\ref{mMAP_d4}) since the computation of the LLR in (\ref{LLR2}) is avoided.

%%%%%%%%
%%%%%%%%
\section{Sub-Optimal Low-Complexity Detectors for DTM}\label{llrapp} 

For STM, we could show that optimal detection can be performed without evaluation of the corresponding LLR in (\ref{LLR2}), which led to the simple optimal detector in (\ref{OptSTMsimple}). However, for DTM, such a simplification does not seem possible and computing the LLR  given in (\ref{TotLLR}) may be computationally too complex for simple nano-machines due to the exponential functions in the numerator and the denominator in (\ref{llr_loc}), which have large dynamic ranges. To illustrate this problem more clearly, we rewrite the optimum $\mathsf{LLR}(\vec{y}^{\{m\}})$ for each sensor $m$ as a function of $\sigma_y^{\{m\}}$ as
\begin{align}\label{newapp}
\mathsf{LLR}_\text{opt}\left(\sigma_y^{\{m\}}\right)=\mathsf{log}\left(\frac{\sum_{x_m \in{\mathcal{X}}}\g_1\left(x_m\right)\text{V}(x_mA+J,\sigma_y^{\{m\}})}{\sum_{x_m \in{\mathcal{X}}}\g_0\left(x_m\right)\text{V}(x_mA+J,\sigma_y^{\{m\}})}\right),
\end{align}
where $\text{V}(a,b) = \Ex{\left(-Na\right)}a^b$. Although the values of $\mathsf{LLR}_\text{opt}\left(\sigma_y^{\{m\}}\right)
$ in (\ref{newapp}) may have a reasonable  range, the term $\text{V}(x_mA+J,\sigma_y^{\{m\}})$ inside the logarithm  can assume very large values if the number of received molecules $\sigma_y^{\{m\}}$ is large. In fact, for computation of the optimal LLR, \textit{we have to first compute all summands} in the numerator and denominator, respectively, then compute the summation for both the numerator and denominator, and \textit{then apply the logarithm} to the ratio of the numerator and denominator. Because of the summations in the numerator and denominator in (\ref{newapp}), the logarithm does not cancel out the exponential functions. This leads to a large dynamic range, especially for large $M$, $N$, and $A$.

In this section, we derive several simple low-complexity sub-optimal detectors for DTM by using appropriate approximations for the LLR. As will be shown in Section~VI, these sub-optimum detectors provide a favorable trade-off between detection performance and complexity, and hence, are more suitable for implementation in nano-machines with limited computational capabilities.
%%%%
%%%%

\subsection{Max-Log Approximation}

This approximation, first introduced in \cite{Robert_MaxLog} for classical wireless sensor networks, takes into account only the maximum terms in both the numerator and denominator of $\mathsf{LLR}_{\text{opt}}$ in (\ref{newapp}). That is,
\begin{align}\label{max-log} %{newapp}
\mathsf{LLR}_{\text{Max-Log}}\left(\sigma_y^{\{m\}}\right)=\mathsf{log} \left( \frac{\displaystyle\max_{x_m \in{\mathcal{X}}}{(\g_1(x_m)\text{V}(x_mA+J,\sigma_y^{\{m\}})})}{\displaystyle\max_{x_m \in{\mathcal{X}}}{(\g_0(x_m)\text{V}(x_mA+J,\sigma_y^{\{m\}})})} \right).
\end{align}
Since the $\mathsf{log}(\cdot)$ function is monotonically increasing, we have $\max_{x}f(x)=\Ex(\max_{x}\mathsf{log} (f(x)))$. Applying this rule to (\ref{max-log}), we obtain
\iffalse
\begin{align}\label{max-log_simple} 
\mathsf{LLR}_{\text{Max-Log}}\left(\sigma_y^{\{m\}}\right)&= 
\frac{\displaystyle\max_{x_m \in{\mathcal{X}}}{\mathsf{log}\left(\g_1(x_m)\text{V}(x_mA+J,\sigma_y^{\{m\}})\right)}}{\displaystyle\max_{x_m \in{\mathcal{X}}}{\mathsf{log}\left(\g_0(x_m)\text{V}(x_mA+J,\sigma_y^{\{m\}})\right)}} \nonumber \\
&=\frac{\max_{x_m \in{\mathcal{X}}}{\mathsf{log}(\g_1(x_m))+\sigma_y^{\{m\}}\mathsf{log}(x_mA+J)-N(x_mA+J)}}
{\max_{x_m \in{\mathcal{X}}}{\mathsf{log}(\g_0(x_m))+\sigma_y^{\{m\}}\mathsf{log}(x_mA+J)-N(x_mA+J)}}.
\end{align}\fi
%
\begin{align}\label{max-log_simple} 
\mathsf{LLR}_{\text{Max-Log}}\left(\sigma_y^{\{m\}}\right)&= 
\displaystyle\max_{x_m \in{\mathcal{X}}}\left( {\mathsf{log}\left(\g_1(x_m)\text{V}(x_mA+J,\sigma_y^{\{m\}})\right)}\right) \nonumber \\
&\quad - \displaystyle\max_{x_m \in{\mathcal{X}}}\left( {\mathsf{log}\left(\g_0(x_m)\text{V}(x_mA+J,\sigma_y^{\{m\}})\right)}\right) \nonumber \\
&=\max_{x_m \in{\mathcal{X}}} \left( {\mathsf{log}(\g_1(x_m))+\sigma_y^{\{m\}}\mathsf{log}(x_mA+J)-N(x_mA+J)}\right) \nonumber \\
&\quad -\max_{x_m \in{\mathcal{X}}} \left( {\mathsf{log}(\g_0(x_m))+\sigma_y^{\{m\}}\mathsf{log}(x_mA+J)-N(x_mA+J)}\right).
\end{align}
As can be seen from (\ref{max-log_simple}), unlike the optimal $\mathsf{LLR}_\text{opt}\left(\sigma_y^{\{m\}}\right)$ in (\ref{newapp}), the expression for the sub-optimal $\mathsf{LLR}_{\text{Max-Log}}\left(\sigma_y^{\{m\}}\right)$ is a linear function of $\sigma_y^{\{m\}}$ which facilitates its numerical evaluation. The sub-optimal detector for the max-log approximation is obtained from (\ref{mMAP_d4}) by substituting $\mathsf{LLR}_{\text{Max-Log}}$ for $\mathsf{LLR}\left( \vec{y}^{\{m\}}\right)$ in (\ref{TotLLR}).
%%%%
%%%%
\subsection{Maximum Ratio Combining (MRC) Approximation} \label{mrc_label}

For the MRC approximation \cite{MRC_Fusion}, all sensors are assumed to be ideal,  i.e.,  $\g_0(x_m)=\delta(x_m)$ and $\g_1(x_m)=\delta(x_m-1)$. By substituting the ideal sensing distributions into (\ref{newapp}), we obtain
\begin{align}\label{MRC}
\mathsf{LLR}_{\text{MRC}}\left(\sigma_y^{\{m\}}\right)
=\mathsf{log}\left(
\frac{\Ex(-N(A+J)) (A+J)^{\sigma_y^{\{m\}}}}{\Ex(-NJ) J^{\sigma_y^{\{m\}}}}
\right) = -NA+\sigma_y^{\{m\}}\mathsf{log}\left(1+\frac{A}{J}\right).
\end{align}
The sub-optimal $\mathsf{LLR}_{\text{MRC}}\left(\sigma_y^{\{m\}}\right)$ is in fact a linear function in $\sigma_y^{\{m\}}$ which again facilitates its numerical evaluation. In addition, applying the MRC approximation given in (\ref{MRC}) in (\ref{mMAP_d4}), the detector  essentially compares the summation of all molecules received from all sensors in all time slots with a suitable threshold, denoted by $\gamma_{\text{MRC}}$. This leads to the following simple detector
\begin{align} \label{dmrc}
\text{d}_{\text{MRC}}=\left\{ 
\begin{array}{l l}
0, & \,\, \text{if $~ \sum_{m=1}^{M} \sigma_y^{\{m\}} \leqslant \gamma_{\text{MRC}}$} \\
1, & \,\, \text{otherwise} \\
\end{array} \right.
\end{align}
The MRC detector is expected to perform well when the sensors are reliable but the reporting channel is unreliable.

\subsection{Chair-Varshney (CV) Approximation}

For the CV approximation, which was first introduced in \cite{CV_Fusion} for classical sensor networks, instead of substituting the numerator and denominator of the LLR in (\ref{newapp}) with their corresponding maximum terms as in the Max-Log decision rule, we first find the value of $x_m$ in (\ref{newapp}) that maximizes only the Poisson terms in the numerator and denominator denoted by~$\hat{x}_m$.
This leads~to
\begin{align}\label{CV-Simple}
\hat{x}_m = \underset{x_m \in{\mathcal{X}}}{\argmax} {\left(\text{V}\left(x_mA+J,\sigma_y^{\{m\}}\right)\right)}
\overset{(a)}{=} \underset{x_m \in{\mathcal{X}}}{\argmax} {-N(x_mA+J)+\sigma_y^{\{m\}}\mathsf{log}(x_mA+J)},
\end{align}
where for equality $(a)$, we exploited that the $\mathsf{log}(\cdot)$ function is monotonically increasing.
Then, the LLR is approximated as
\begin{align}\label{CV}
\mathsf{LLR}_{\text{CV}}\left(\sigma_{y}^{\{m\}}\right)=\mathsf{log}\left(\frac{\g_1(\hat{x}_m)}{\g_0(\hat{x}_m)}\right).
\end{align}
By substituting the above approximated LLR into (\ref{TotLLR}), the corresponding sub-optimal detector is given by (\ref{mMAP_d4}).  The CV detector is expected to perform well in situations where the reporting channel is reliable but the sensors are unreliable.

%where
%\begin{align} \label{CVmidd} 
%\hat{x}_m = \left\{
% \begin{array}{ll}
% 0, & \mbox{if } \,\, \frac{\sigma_y^{\{m\}}}{N} < J \\
% \text{I}(\sigma_y^{\{m\}}), \,\,& \mbox{if } \,\, J \leq \frac{\sigma_y^{\{m\}}}{N} \leq A+J\\
% 1, & \mbox{if } \,\, \frac{\sigma_y^{\{m\}}}{N} > A+J,
% \end{array}
%\right.
%\end{align}
%with
%\begin{align} \label{two} 
%\text{I}(\sigma_y^{\{m\}}) = \left\{
% \begin{array}{ll}
% \frac{\sigma_y^{\{m\}}/N-J}{A}, & \mbox{if } ~ \frac{\sigma_y^{\{m\}}/N-J}{A} \in \mathcal{X} \\
% \max \left( \text{V}\left(x_m^{*}A+J,\sigma_y^{\{m\}}\right), \text{V}\left(x_m^{**}A+J,\sigma_y^{\{m\}}\right) \right), & \mbox{otherwise}
% \end{array}
%\right.
%\end{align} %
%and 
%\begin{align}
%x_m^{*} = \frac{\lfloor (\sigma_{y}-NJ)(L-1)/(NA)\rfloor} {L-1}, ~x_m^{**} = x_m^{*}+\frac{1}{L-1}.
%\end{align}

\subsection{Two-Stage Detector}\label{ids}
In this part, we consider a simple sub-optimum detector, where the decision is made in two stages. In the first stage, the FC makes a separate decision on the observations received from each sensor.   Then, in the second stage, the global final decision is made using the $M$ individual decisions. Our motivation for considering this sub-optimum detector is that it is simple and it may be more practical when the adopted FC is a biological entity. In particular, for some biological cells, the processing unit of a cell  does not have direct access to the counted molecules. Rather, a preliminary decision may be made by the signaling pathway which connects the receptors to the processing unit of the cell \cite{CellBio}.

For each sensor $m$, the detector makes a local decision by comparing $\mathsf{LLR}(\vec{y}^{\{m\}})$, given in (\ref{llr_loc}), with threshold level $\tau$, that is
\begin{align} \label{mMAP_d111}
\text{d}_{\text{local}}^{\{m\}}=\left\{ 
\begin{array}{l l}
0,\,\, &  \text{if $~ \mathsf{LLR}(\vec{y}^{\{m\}}) \leqslant \tau$} \\
1, &  \text{otherwise.} \\
\end{array} \right.
\end{align} 
Using Theorem~\ref{Th2}, the optimum local decision rule in (\ref{mMAP_d111}) can be further simplified for ease of implementation.
%***
In particular, in Theorem~\ref{Th2}, we showed that $\mathsf{LLR}(\vec{y}^{\{m\}})$ is strictly increasing in terms of $\sum_{n=1}^{N}y^{\{m\}}_n$. That is, $\mathsf{LLR}(\vec{y}^{\{m\}})$ is a monotonic function in $\sum_{n=1}^{N}y^{\{m\}}_n$.
This property implies that, comparing $\sum_{n=1}^{N}y^{\{m\}}_n$ with an appropriate new threshold, denoted by $\gamma_{\text{local}}$, is equivalent to comparing the LLR, $\mathsf{LLR}(\vec{y}^{\{m\}})$, with the original threshold, $\tau$. 
Therefore, the optimal local decision can be simplified to
\begin{align} \label{mMAP_d}
\text{d}_{\text{local}}^{\{m\}}=\left\{ 
\begin{array}{l l}
0 & \quad \text{if $~\sum_{n=1}^{N} y_n^{\{m\}} \leqslant \gamma_{\text{local}}$} \\
1 & \quad \text{otherwise}. \\
\end{array} \right.
\end{align} 

In the second stage, the detector makes the final decision based on the number of sensors that have made a local decision in favor of $\HHH$. If the number of sensors with $\text{d}_{\text{local}}^{\{m\}}=1$ is larger than a new threshold, $\gamma_{\text{global}}$, the final decision is $\HHH$, otherwise it is $\HH$. This leads to
\begin{align} \label{mMAP_dglob}
\text{d}_{\text{final}}=\left\{ 
\begin{array}{l l}
0,\,\, &  \text{if $~\sum_{m=1}^{M} \text{d}_{\text{local}}^{\{m\}} \leqslant \gamma_{\text{global}}$} \\
1, &  \text{otherwise}. \\
\end{array} \right.
\end{align} 
%
%%****************************
We note that the above two-stage detector has been also considered in the context of spectrum sensing in cognitive radio networks and is there referred to as ``hard decision"  with ``$K$-out-of-$M$" decision rule where $K=\gamma_{\text{global}}-1$ \cite{SpectrumSensingSurvey}. Special cases of the $K$-out-of-$M$ decision rule are the ``OR" rule where $K=1$ and the ``AND" rule where $K=M$.

\subsection{Accuracy of Approximation}\label{Complexity}

The proposed sub-optimal LLRs for DTM are computationally much simpler compared to the optimal LLR. However, it is expected that this simplicity is accompanied by a loss in performance. Hence, in the following, we study the accuracy of the proposed LLR approximations for a simple example.
Fig.~\ref{Fig:Accuracy} shows $\mathsf{LLR}_\text{opt}\left(\sigma_{y}^{\{m\}}\right)$ and its approximations, namely $\mathsf{LLR}_\text{Max-Log}\left(\sigma_{y}^{\{m\}}\right)$, $\mathsf{LLR}_\text{MRC}\left(\sigma_{y}^{\{m\}}\right)$, and $\mathsf{LLR}_\text{CV}\left(\sigma_{y}^{\{m\}}\right)$, versus $\sigma_{y}^{\{m\}} = \sum_{n=1}^{N}y_n^{\{m\}}$ for $L=4$, $J=4$, $N=10$, and $A=6$. While the optimum LLR is generally non-linear, all approximated LLRs are piecewise linear. Among the different approximations, the Max-Log approximation follows the optimum LLR most closely. Thus, we expect that the Max-Log detector outperforms the other proposed sub-optimal detectors. More detailed discussions regarding the relative performance of the proposed detectors are provided in Section~VI.

\begin{figure}
\centering
\includegraphics[width=0.6\linewidth]{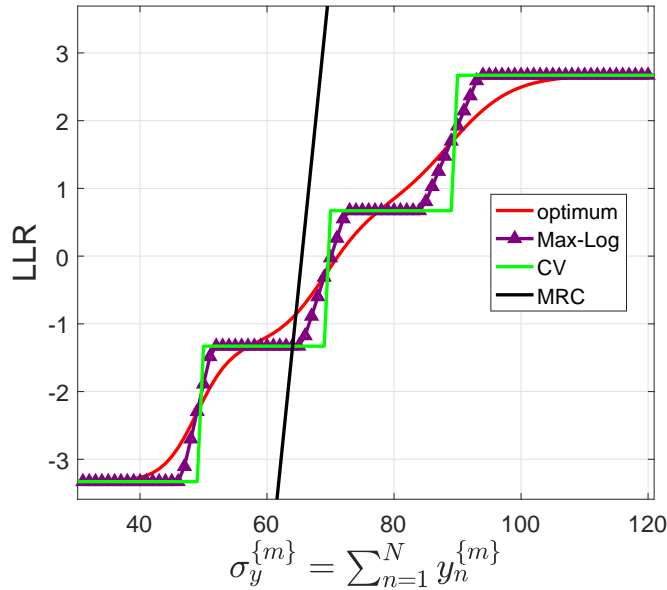} %
\vspace{-0.1cm}
\caption{\footnotesize{LLR for DTM versus $\sigma_y^{\{m\}}=\sum_{n=1}^{N}y_n^{\{m\}}$, for parameters $L=4, J=4$, $N=10$, and $A=6$. }}
\label{Fig:Accuracy}
\vspace{-0.3cm}
\end{figure}

\section{Performance Analysis}\label{sys1p3}
In this section, we first present a numerical approach for evaluation of the performance of DTM and STM.
Moreover, we show that the decision rules of the optimal detector for STM and the sub-optimal MRC and two-stage detectors for DTM allow for a closed-form analytical performance evaluation.
Finally, we present an asymptotic performance bound for DTM and STM.%~\ref{asympsec}.
 
For performance analysis, we study the false alarm probability, denoted by $\Pf$, and the detection probability, denoted by $\Pd$. By expressing the final decision in terms of the total LLR and denoting the threshold of the decision rule by $\gamma$, the false alarm and detection probabilities can be calculated as follows
\begin{align}\label{integral_in_TH2_1}
\Pf = \sum_{\ell > \gamma}\text{P}_{\overline{\mathsf{LLR}}}(\ell\vert \HH)  \quad
\text{and} \quad
\Pd = \sum_{\ell > \gamma}\text{P}_{\overline{\mathsf{LLR}}}(\ell\vert \HHH) ,
\end{align}
where $\text{P}_{\overline{\mathsf{LLR}}}(\ell|\HTemp_i)$ is the probability mass function (PMF) of the LLR for hypothesis $\HTemp_i$ and $\ell$ is its argument. Note that since the observations, i.e., the numbers of received molecules, are integers, the LLR takes discrete values. In the following, we present a numerical performance evaluation method for DTM and STM.
 
\subsection{DTM Reporting}

\subsubsection{General Analysis}
In the following theorem, we employ the method introduced in \cite{Robert_MaxLog} to numerically analyze the probabilities of detection and false alarm for any LLR-based decision metric of the form $\sum_{m=1}^{M} \mathsf{LLR}(\cdot)$. Note that the decision metrics of all considered optimal and sub-optimal DTM detectors can be expressed in this manner, except for the sub-optimal two-stage~detector. 
 %*
\begin{theorem}\label{theorem_general_Pfa_d}
Given the sensing distributions $\g_i\left(x_m\right), x_m\in{\mathcal{X}}$, for hypothesis $\HTemp_i, i=0, 1$, and any decision rule of the form $\sum_{m=1}^{M}\mathsf{LLR}\left(\sigma_{y}^{\{m\}}\right)$, where 
$\sigma_{y}^{\{m\}}$ is the argument of the local LLR, $\mathsf{LLR}(\cdot)$,
 corresponding to the $m$-th sensor, we have
\begin{align}
\Pf =\sum_{\vec{\sigma}_y \in \mathcal{S}}~\prod_{m=1}^M
\mathsf{L}_0\left(\sigma_{y}^{\{m\}}\right) \quad \text{and} \quad
\Pd =\sum_{\vec{\sigma}_y \in \mathcal{S}}~\prod_{m=1}^M
\mathsf{L}_1\left(\sigma_{y}^{\{m\}}\right), \label{Pfad_difinition}
\end{align}
where $\mathcal{S}\triangleq\Big\{ \vec{\sigma}_y = [\sigma_{y}^{\{1\}}, \cdots, \sigma_{y}^{\{M\}}]:~\sum_{m=1}^{M}\mathsf{LLR}\left(\sigma_{y}^{\{m\}}\right) > \gamma \Big\}$ and $\mathsf{L}_i(\cdot), i=0,1,$ is~defined~as
\begin{align}\label{T_h_difinition}
\mathsf{L}_i\left(w\right) \triangleq \sum_{x_m \in{\mathcal{X}}}\text{P}\left(w\vert x_m\right)\g_i\left(x_m\right),
\end{align}
with $\text{P}(w\vert x_m)$ being the probability that the FC receives $\sigma_{y}^{\{m\}}=w$ molecules from sensor $m$, when sensor $m$ observes $x_m$. 
\begin{proof}
The proof is given in Appendix B. %\ref{Th1:Proof}.
\end{proof}
\end{theorem}
%*

The result in (\ref{Pfad_difinition}) can be used to plot the Receiver Operating Curve (ROC) of the DTM detectors (except the two-stage detector) after substituting the corresponding LLR in set $\mathcal{S}$. In particular, using (\ref{Hi-new}), $\text{P}(\sigma_{y}^{\{m\}}\vert x_m)$ required in Theorem \ref{theorem_general_Pfa_d} is given by
\begin{align}
\text{P}(\sigma_{y}^{\{m\}}\vert x_m)=\frac{\Ex\left(-N(x_mA+J)\right)\left(x_mA+J\right)^{\sigma_{y}^{\{m\}}}}{\sigma_{y}^{\{m\}}!}.
\end{align}
Combining this result with (\ref{T_h_difinition}), we obtain
\begin{align}\label{optimum}
\mathsf{L}_i(\sigma_{y}^{\{m\}})=
\sum_{x_m \in{\mathcal{X}}}\g_i(x_m)\frac{\Ex\left(-N(x_mA+J)\right)\left(x_mA+J\right)^{\sigma_{y}^{\{m\}}}}{\sigma_{y}^{\{m\}}!},\quad i=0,1.
\end{align}

\subsubsection{Two-Stage Detector} 
 \label{ids_res}
For the $m$-th sensor, based on the local decision rule in (\ref{mMAP_d}), we can derive the local false alarm probability as
\begin{align}\label{pflmy}
\Pfl   &=  \msP \left( \sum_{n=1}^{N} y_n^{\{m \}} > \gamma_{\text{local}} |\HH \right)
= \sum_{x_m \in \mathcal{X}}\msP \left( \sum_{n=1}^{N} y_n^{\{m \}} > \gamma_{\text{local}} \Big | X^{\{ m\}}=x_m, \HH\right)
\nonumber \\
&\quad\times \msP \left( X^{\{ m\}}=x_m |\HH \right) 
= \sum_{x_m \in \mathcal{X}} \g_{0}(x_m) \text{H}(\gamma_{\text{local}}, NJ+Nx_mA),
\end{align}
where $\text{H}(x,\lambda)$ is the
complementary commutative distribution function (CDF) of the Poisson distribution, i.e., 
\begin{align}\label{Cum}
\text{H}(x,\lambda) = \sum_{k=x+1}^{+\infty} \frac{\Ex(-\lambda) \lambda^k}{k!}.
\end{align}
Similarly the local detection probability can be derived as follows
\begin{align}\label{pdlmy}
\Pdl  =  \msP \left( \sum_{n=1}^{N} y_n^{\{m \}} > \gamma_{\text{local}} |\HHH \right) 
		= \sum_{x_m \in \mathcal{X}} \g_{1}(x_m) \text{H}(\gamma_{\text{local}}, NJ+Nx_mA).
\end{align}
As discussed in Section \ref{ids}, the overall decision rule compares the number of sensors that decide for $\HHH$ with an integer threshold $\gamma_{\text{global}}$. With the decision rule given in (\ref{mMAP_dglob}), the overall false alarm and detection probabilities are obtained as
\begin{align}
 \Pf     &=  \msP \left( \sum_{m=1}^{M}  \text{d}_{\text{local}}^{\{m\}} > \gamma_{\text{global}}  |\HH \right) = \sum_{m=\gamma_{\text{global}}+1}^{M} {M \choose m} \left(\Pfl\right) ^{m} (1- \Pfl)^{M-m} \quad \text{and}\\
 \Pd 	  &=  \msP \left( \sum_{m=1}^{M}  \text{d}_{\text{local}}^{\{m\}} > \gamma_{\text{global}}  |\HHH \right) = \sum_{m=\gamma_{\text{global}}+1}^{M} {M \choose m} \left(\Pdl\right) ^{m} (1- \Pdl)^{M-m},
\end{align}
respectively, where $\Pfl$ and $\Pdl$ are given in (\ref{pflmy}) and (\ref{pdlmy}), respectively.

\subsection{STM Reporting}

For STM, we can adopt $\sigma_y=\sum_{m=1}^{M} \sum_{n=1}^{N} y_n^{\{m\}}$ as the decision variable instead of the LLR, see Corollary~\ref{CorolSTM} and (\ref{OptSTMsimple}). This facilitates the derivation of closed-form expressions for the false alarm and detection probabilities. In particular, the false alarm probability can be derived as
\begin{align}\label{pf11}
\Pf = \msP \left( \sigma_y > \gamma |\HH \right)
&=\sum_{x\in \mathcal{M}'}\msP \left(  \sum_{m=1}^{M} \sum_{n=1}^{N} y_n^{\{m \}}>\gamma  \bigg |  \sum_{m=1}^{M}X^{\{ m\}}=x,\HH   \right)
\msP \left( \sum_{m=1}^{M}X^{\{ m\}}=x \big| \HH \right) \nonumber \\
&= \sum_{x\in \mathcal{M}'}   \GG_0(x)  \text{H}(\gamma, NJ+xNA).
\end{align}
Similarly, the detection probability can be computed as
\begin{align}\label{pd11}
\Pd  =\msP \left( \sigma_y > \gamma |\HHH \right) 
%&= \sum_{k=0}^{m}\msP \left(  \sum_{j=1}^{m} \sum_{i=1}^{n} y_i^{\{j \}}>\gamma, \sum_{j=1}^{m}Z_1^{\{ j\}}=k  | \HHH   \right) \nonumber \\
=\sum_{x\in \mathcal{M}'} \GG_1(x)  \text{H}(\gamma, NJ+xNA).
\end{align}

\begin{remark}
The performance of the MRC detector for DTM can be evaluated based on Theorem~\ref{theorem_general_Pfa_d}. Alternatively, since the MRC detector for DTM also employs $\sigma_y=\sum_{m=1}^{M} \sum_{n=1}^{N} y_n^{\{m\}}$ as decision variable, the closed-form expressions for the false alarm and detection probabilities for STM given in (\ref{pf11}) and (\ref{pd11}), respectively, are also valid for the MRC detector for DTM, when $J$ is replaced by $MJ$.
\end{remark}

%

%%******-------------------------------------
\subsection{Asymptotic Performance} \label{asympsec}

In this subsection, we study the asymptotic performance with respect to the number of deployed sensors.
 As an asymptotic performance metric, we investigate the exponential rate of decay of the probability of missed detection, i.e., $\Pm= 1-\Pd$, and the probability of false alarm, i.e.,~$\Pf$, as $M\to\infty$.

For simplicity, let us define $ \overline{\mathsf{LR}}(\vec{\sigma}_y)=\mathsf{exp}(\overline{\mathsf{LLR}}(\vec{\sigma}_y))$ and $\mathsf{LR}(\sigma_y^{\{m\}})=\mathsf{exp}(\mathsf{LLR}(\sigma_y^{\{m\}}))$ as the likelihood ratio of the total observation and the observation from sensor $m$ at the FC, respectively. Note that due to our assumption that  observations are i.i.d.,  $\overline{\mathsf{LR}}(\vec{\sigma}_y)=\prod_{m=1}^M \mathsf{LR}(\sigma_y^{\{m\}})$ holds. Using the Chernoff bound \cite{Chernoff_Bound}, we obtain the following bounds on the probabilities of false alarm and missed detection
\begin{align}\label{PfPmChernof}
\Pf &= \msP \left(\mathsf{log} \left( \overline{\mathsf{LR}}(\vec{\sigma}_y) \right)-\gamma >0| \HH \right) \leq \Ex \left(  \mu_0(s)-s \gamma \right),~\forall s>0 \nonumber \\
\Pm &= \msP \left(\mathsf{log} \left( \overline{\mathsf{LR}}(\vec{\sigma}_y) \right)-\gamma <0| \HHH \right) \leq \Ex \left(  \mu_1(s)-s \gamma \right),~ \forall s<0, 
\end{align}
respectively, where $\mu_i(s), i=0,1,$ is the logarithm of the moment generating function (MGF) of the decision variable, i.e., $\mathsf{log} \left( \overline{\mathsf{LR}}(\vec{\sigma}_y) \right)$, under hypothesis $\HTemp_i$, i.e.,
\begin{align}\label{Ex}
\Ex(\mu_i(s)) = \E \left(  \Ex\left(s  \mathsf{log} \left( \overline{\mathsf{LR}}(\vec{\sigma}_y) \right)\right) |\HTemp_i \right) 
					 =\E \left(  \left( \overline{\mathsf{LR}}(\vec{\sigma}_y) \right)^s|\HTemp_i  \right),\quad i=0,1.
\end{align}
By substituting (\ref{Ex}) into (\ref{PfPmChernof}) and exploiting relation $\overline{\mathsf{LR}}(\vec{\sigma}_y)=\prod_{m=1}^M \mathsf{LR}(\sigma_y^{\{m\}})$, the upper bounds become
\begin{align}
\Pf &\leq \Ex \left( M \mathsf{log} \left( \sum_{x\in \mathbb{Z}} \left(\mathsf{LR}(x)\right)^s f_0(x) \right) -s \gamma \right)\triangleq \Pf^{\text{upp}}(s),~\forall s>0 \label{Bound111} \\
\Pm &\leq \Ex \left( M \mathsf{log} \left( \sum_{x\in \mathbb{Z}} \left(\mathsf{LR}(x)\right)^s f_1(x)\right) - s \gamma \right)\triangleq \Pm^{\text{upp}}(s),~ \forall s<0,  \label{Bound222}
\end{align}
where for simplicity of presentation, we have denoted $\msP( \sigma_y^{\{m\}}=x | \HTemp_i)$ by $f_i(x), i=0, 1$. We note that asymptotically as $M\to\infty$, we obtain  
\begin{align}
\Pf^{\text{upp}}(s)\approx \mathsf{exp}(-M\text{Ex}_0 )\quad \text{and} \quad \Pm^{\text{upp}}(s)\approx \mathsf{exp}(-M\text{Ex}_1 ), \label{UpperEq}
\end{align}
where 
\begin{align}
\text{Ex}_i = -\mathsf{log} \left( \sum_{x\in \mathbb{Z}} \left(\mathsf{LR}(x\right)^s f_i(x) \right),\,\,i=0,1. \label{Exponent_i}
\end{align}
Therefore, exponents $\text{Ex}_0$ and $\text{Ex}_1$ reveal how $\Pf$ and $\Pm$ improve, respectively, as the number of sensors increases.

Note that for any positive $s$, $\Pf^{\text{upp}}(s)$ and $\Pm^{\text{upp}}(s)$ constitute upper bounds for $\Pf$ and $\Pm$, respectively. Therefore, the upper bounds can be tightened by optimizing $\Pf^{\text{upp}}(s)$ and $\Pm^{\text{upp}}(s)$ with respect to $s$, i.e.,
\begin{align}\label{Eq:PeOpt}
s_0^* = \underset{s>0}{\mathsf{argmin}} \,\, \Pf^{\text{upp}}(s) \quad \text{and} \quad
s_1^* = \underset{s<0}{\mathsf{argmin}} \,\, \Pm^{\text{upp}}(s),
\end{align}
yield the tightest upper bounds. Moreover, since $\Pf^{\text{upp}}(s)$ and $\Pm^{\text{upp}}(s)$ have an exponential form and an exponential function is monotonically increasing in terms of its argument, one can optimize the argument of the exponential function instead of directly optimizing the function as in (\ref{Eq:PeOpt}). The optimal $s_i^*$ belongs to one of the the stationary points of  $\Pf^{\text{upp}}(s)$ and $\Pm^{\text{upp}}(s)$ which are obtained by taking the derivatives of $\mathsf{log}(\Pf^{\text{upp}}(s))$ and $\mathsf{log}(\Pm^{\text{upp}}(s))$ with respect to $s$. This leads to the following equation
\begin{align}\label{s_opt}
\frac{\mathsf{d} \mathsf{log}\left({\text{P}}_{t}^{\text{upp}}(s) \right) }{\mathsf{d} s} & = \frac{M\sum_{x\in \mathbb{Z}} s\left(\mathsf{LR}(x)\right)^{s-1} f_i(x)}{\sum_{x\in \mathbb{Z}} \left(\mathsf{LR}(x)\right)^s f_i(x)} = ~\gamma~~ \text{for}~ i=0,1, ~\text{and} ~t=\text{fa},~\text{m}%\begin{cases}
%\eta,\,\,&\mathrm{if}\,\,\,i=0\,\text{and}\,\,t=\text{fa} \\
%\eta,\,\,&\mathrm{if}\,\,\,i=1\,\text{and}\,\,t=\text{m}
%\end{cases}
 \nonumber \\
& \overset{M\to\infty}{\longrightarrow} \quad
\sum_{x\in \mathbb{Z}} \left(\mathsf{LR}(x)\right)^{s-1} f_i(x) = 0.
\end{align}

The above equation can be solved numerically using mathematical software packages such as Mathematica.  The tightness of the derived bounds and approximations when utilizing the optimum value of $s$ as the solution to (\ref{s_opt}) will be confirmed in the next section.

\section{Numerical and Simulation Results}
In this section, we provide numerical and simulation results to assess the
performance of the detectors proposed in this paper.

%%******-------------------------------------
\subsection{Simulation Setup}
We define the signal-to-noise ratio as $\text{SNR}=A/J$, where given the definition of $A = \sum_{k=0}^{\infty}h_k A^{\max}$, the SNR is in fact  the expected number of molecules received at the FC due to emission of the $A^{\max}$ molecules by one sensor  divided by the expected number of noise molecules observed in one time slot. 
Next, we present the distributions $\g_0(x_m)$ and $\g_1(x_m)$ of the $L$ possible sensing values. We expect that under hypothesis $\HH$ ($\HHH$), lower (higher) values for $x_m$ are more probable. That is, in the absence of the abnormality, a smaller number of molecules is expected to be released from the sensors compared with the case that the abnormality exists. Here, we adopt distributions of the form $a_i\Ex(b_ix)$ which satisfy the aforementioned condition with  appropriately chosen values for the constants $a_i$ and $b_i$, i.e.,
 \begin{align} %\g_0(x_m) = \frac{1-\Ex\left(-\frac{2.5x_m}{L-1}\right)}{1-\Ex\left(-\frac{2.5}{L-1}\right)}\Ex\left(-2.5x_m \right), ~~
 \g_0(x_m) = \frac{\Ex\left(-2.5x_m \right)}{\sum_{x \in \mathcal{X}}\Ex\left(-2.5x \right)} , ~~
\g_1(x_m) = \frac{\Ex\left(3.5x_m \right)}{\sum_{x \in \mathcal{X}}\Ex\left(3.5x \right)}.
\end{align} 
Note that $a_i$ is a normalization constant which ensures that the sum of all probabilities $\g_i(x_m)$ over $x_m\in\mathcal{X}$ is one, i.e., $\sum_{x_m\in\mathcal{X}}\g_i(x_m)=1$. Moreover, larger values for $|b_i|$ correspond to distributions that imply more reliable sensors. 
%For the hard decision case, where the sensors make binary decisions \color{red}WHAT v\color{black} the presence of abnormality, the sensing errors, i.e., $p_0$ and $p_1$ defined in Section II. 
In order to have a fair comparison between a hard decision scheme, defined by $L=2$ and probabilities $p_0$ and $p_1$, and a soft decision scheme, defined by $L>2$ and PDFs $\g_0(x_m)$ and $\g_1(x_m)$,  we choose $p_0$ and $p_1$ as functions of $\g_0(x_m)$ and $\g_1(x_m)$, respectively, as follows
\begin{align}
p_0=0.5 \g_0\left(0.5\right)+\sum_{x_m>0.5}\g_0\left(x_m\right),~~
p_1=0.5 \g_1\left(0.5\right)+\sum_{x_m<0.5} \g_1\left(x_m\right).
\end{align}

\iffalse
Finally, to have fair comparisons of hard and soft decision schemes, the average number of molecules released for these two cases are set equal, by considering different maximum transmitted molecules, i.e., parameter $A$, for two cases, as $\text{A}_{\text{hard}}$ and $\text{A}_{\text{soft}}$. These two maximum parameters must satisfy the following equality
\begin{align}\label{equi}
\text{A}_{\text{hard}}\left(1-p_1+p_0\right)=\text{A}_{\text{soft}}\sum_{x_m\in{\mathcal{X}}} x_m\left(\g_0\left(x_m\right)+\g_1\left(x_m\right)\right)
\end{align}
where equiprobable hypotheses is pre-assumed. 
\fi
Furthermore, for all results, we assume an independent non-zero environmental Poisson noise with mean $J=4$, unless specified otherwise, and adopt a time slot duration of $T=100 \mu s$. The numerical results are based on the analysis provided in Sections~V-A and V-B, and for the simulation results, the probability of error is determined by averaging over $10^6$ independent realizations of $Y_n^{\{m\}}$ and $Y_n$.

\begin{remark}
We note that in \cite{Lahouti_Detection}, the hypothesis testing is composite, i.e., some parameters of the system model are not known. In contrast, in this paper, all parameters of the system model are assumed to be known. In addition,  in \cite{Lahouti_Detection}, the reporting channel is assumed to be an AWGN channel, whereas in this paper, the reporting channel is modeled as a Poisson process which introduces an inherent \textit{signal-dependent} noise, even if there is no background noise ($J_m=J=0,\,\,\forall m$). Therefore, given these fundamental differences, we do not compare our results with those of \cite{Lahouti_Detection},  since such a comparison  would be unfair.
\end{remark}

\subsection{Performance Evaluation}
In the following, we comprehensively study and compare the performance of the proposed detectors and signaling schemes and verify our analysis.

%%%***
\begin{figure}
\centering 
\includegraphics[width=0.6\textwidth] {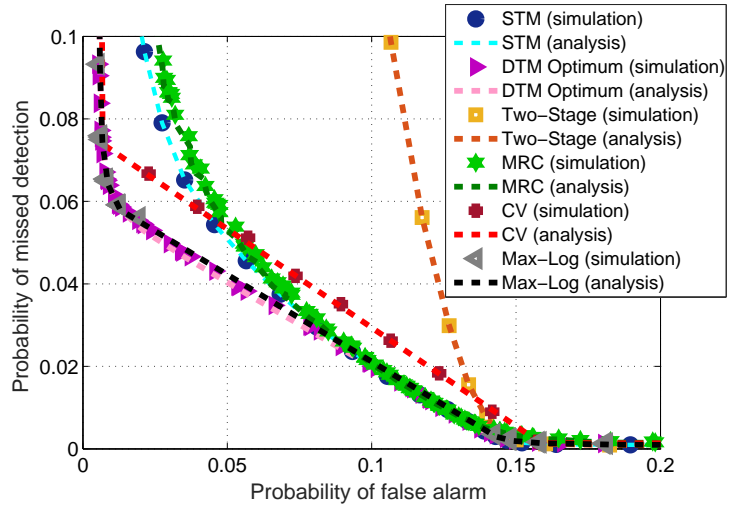} %fig8_n1_m2_J4_A15.eps %AllComparision.eps % fig3
\vspace{-0.6cm} %ig8_n1_m2_J4_A15
\caption{ ROC for parameters $L=2, A=15, J=4, M=2$, and $N=1$. }
\label{Fig:Analysis}
\vspace{-0.75cm}
\end{figure}

%%******-------------------------------------
%\subsubsection{Validity of the Simulation/Analytical Results}
%Fig7
In Fig. \ref{Fig:Analysis}, the analytical results presented in Theorem \ref{theorem_general_Pfa_d} are confirmed by simulations. In particular, Fig. \ref{Fig:Analysis} shows simulation and analytical results for the probability of  missed detection versus the probability of false alarm for different detectors. 
As can be observed, for all detectors, the theoretical results  are in excellent agreement with the simulation results. Moreover, for the set of parameters considered in Fig. \ref{Fig:Analysis}, the optimal detector for DTM outperforms the optimal detector for STM especially for small false alarm probabilities, i.e., for $P_{\text{fa}}\leq 0.08$.  In addition, for DTM, the Max-Log detector closely approaches the performance of the optimal detector and outperforms all other sub-optimal detectors. Furthermore, for the adopted set of system parameters, the CV detector has a superior performance compared to the MRC detector for small false alarm probabilities, i.e., for $P_{\text{fa}}\leq 0.048$. In addition, for false alarm probabilities $P_{\text{fa}}\leq 0.15$, the simple two-stage detector suffers from a considerable performance loss compared to the other detectors; however, as the false alarm probability increases, this performance loss decreases too.

%%******-------------------------------------
%\subsubsection{Comparison of the DTM and STM Schemes}

Next, we compare the performance of the optimal detectors for DTM and STM in more detail. In particular, in Fig. \ref{Fig:DSComp}, the probability of missed detection for DTM and STM is plotted versus $A$ for different values of $J$ and a given false alarm probability of $\Pf=0.05$. We observe from Fig. \ref{Fig:DSComp} that as the value of $A$ increases, the probability of missed detection decreases since the MC reporting channel becomes more reliable. Moreover, we observe that the relative performance of STM and DTM depends on the value $J$. In fact, for small $J$, DTM outperforms  STM  whereas for large $J$, STM may significantly outperform DTM. In fact, the advantage of DTM over STM is that that the FC can distinguish between the molecules released by different sensors and exploit this additional knowledge for the improvement of the detection performance. On the other hand, the advantage of STM over DTM comes from the fact that the mean number of noise molecules for \textit{each type of molecule} is constant which leads to a higher overall noise for DTM compared to STM, especially when the number of sensors is large. Therefore, whether DTM or STM is preferable depends on the system parameters.

%%%***
\begin{figure} 
\centering 
\includegraphics[width=0.6\textwidth] {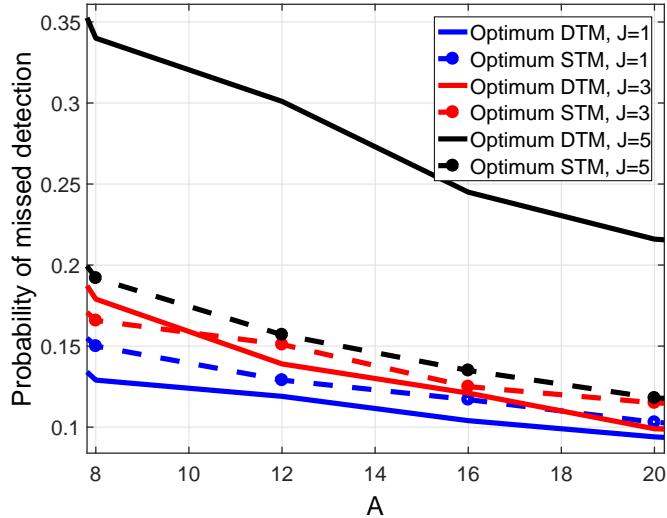}    % or FigComp_.eps or STMvsDTM__.eps 
\vspace{-0.5cm}
\caption{ Probability of missed detection versus $A$ for STM and DTM using the respective optimal decision rules for the parameters of $L=4, M=2, N=2$, and $\Pf = 0.05$.  All curves were obtained via simulation.}
\label{Fig:DSComp}
\vspace{-0.3cm}
\end{figure}

%%******-------------------------------------
%\subsubsection{Further Results for the DTM Scheme}

In the following, we investigate the performance of the proposed sub-optimal low-complexity detectors for DTM in detail. In Fig.~\ref{Fig:ApproxSH}, we show the probability of missed detection versus the probability of false alarm for all proposed DTM detectors and for both hard decision, i.e., $L=2$, and soft decision, i.e., $L=4$, at the sensors. From this figure, it can be observed that the performance of soft decision is superior to that of hard decision for all detectors. In addition, among the sub-optimum detectors, the probability of missed detection for the Max-Log detector is lower than that of the other sub-optimal detectors for any given false alarm probability, which confirms our expectations from the LLR comparison in Fig.~\ref{Fig:Accuracy}. 
We note that the sudden changes in the performance of the CV detector are due to the discontinuity of its LLR.

%%%***
\begin{figure}
\centering
\includegraphics[width=0.6\linewidth]{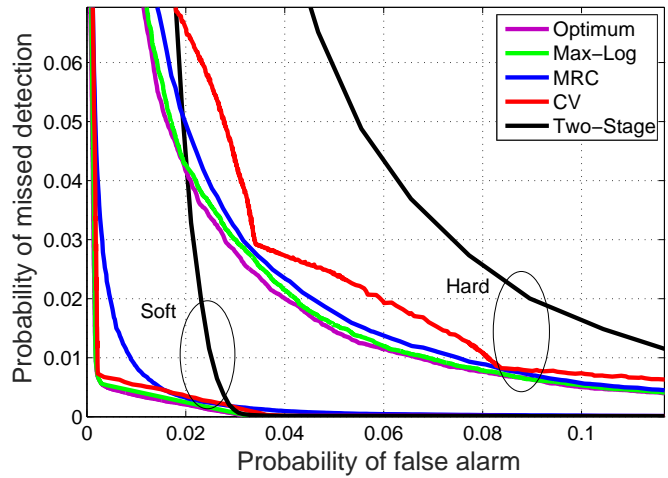} %Soft_Hard2  %DTM_SubOpt_
\vspace{-0.3cm}
\caption{\footnotesize{ Probability of missed detection versus probability of false alarm, for parameters $J=4, M=2, N=2$, and $A=15$ for DTM.  All curves were obtained via simulation.}}
\label{Fig:ApproxSH}
\vspace{-0.3cm}
\end{figure}
%
%%%***
\begin{figure}
\centering
\includegraphics[width=0.6\linewidth]{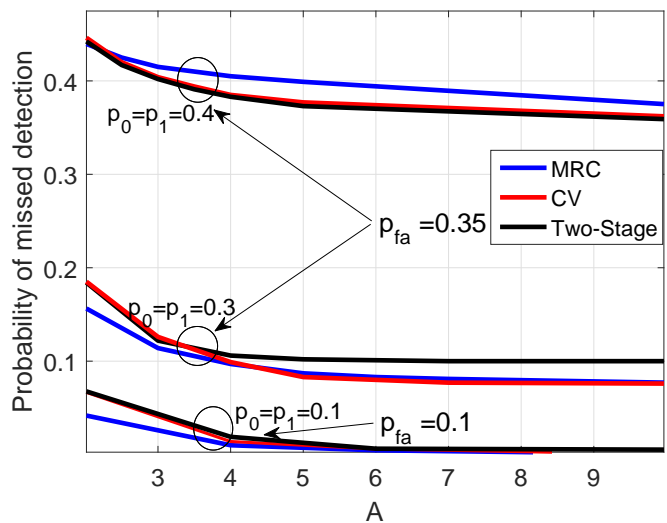} %FigN.eps %New.eps
\vspace{-0.3cm}
\caption{\footnotesize{ Probability of missed detection versus $A$ for parameters $L=2, J=1, M=5, N=2$, $p_0=p_1\in\{0.1,0.3,0.4\}$, and $\Pf\in\{0,1,0.35\}$, and for the two-stage, CV, and MRC detectors.  All curves were obtained via simulation.}}
\label{Fig:New}
\vspace{-0.3cm}
\end{figure}

It is interesting to further study and compare the performances of the CV, MRC, and two-stage detectors since none of these detectors is uniformly better than the other two, cf. Fig.~\ref{Fig:ApproxSH}. 
In Fig. \ref{Fig:New}, we show the probability of missed detection for the CV, MRC, and two-stage detectors versus $A$ for three sensing scenarios, namely $p_0 =p_1=0.4$ and $p_0=p_1=0.3$ with a false alarm probability of $\Pf=0.35$ and $p_0 =p_1=0.1$ with a false alarm probability of $\Pf=0.1$.  
As can be observed,  the performance of the MRC detector improves relative to that of the CV and two-stage detectors as $p_0$ and $p_1$ decrease. This is due to the fact that we assumed ideal sensing to derive the MRC detector, cf. Section IV-B. Therefore, as the sensors becomes more reliable, i.e., for smaller $p_0$ and $p_1$, the MRC approximation of the optimal LLR becomes more accurate. Hence, for sufficiently reliable sensors, the MRC detector outperforms the CV detector whereas for unreliable sensors, the CV detector performs better than the MRC detector, especially for large values of $A$, as an ideal reporting channel was assume for derivation of the CV detector. Despite its simplicity, the two-stage detector also performs well compared to the other two detectors, especially for $p_0 =p_1=0.4$ and $p_0=p_1=0.1$.

Next, we study the effect that parameters $N$ and $M$ have on detection performance. We provide results only for the MRC detector for  clarity of presentation. In particular, from Fig.~\ref{Fig:MN}, we observe that although increasing $N$ improves the detection performance, the impact of increasing the number of sensors is more substantial. 
This is due to the fact that by increasing $N$, effectively, the overall reliability of the MC reporting channel improves,  similar to the effect that repetition codes have in wireless communications, but the overall detection performance is still limited by the reliability of the sensing mechanism at the sensors. In contrast, by increasing $M$, the number of independent sensing observations increases,  similar to effect that increasing the diversity gain has in conventional wireless communications. The latter effect is generally  more substantial compared to the former as far as  the overall detection performance is concerned.
%

%%%***
\begin{figure}
\centering
\includegraphics[width=0.6\linewidth]{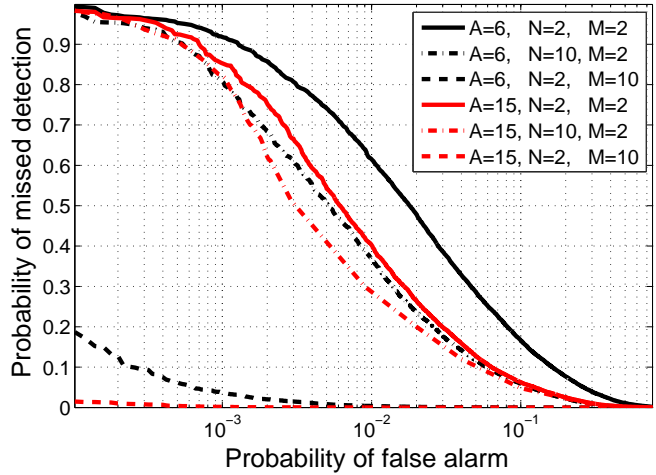} %DTM_A_N_M
\vspace{-0.3cm}
\caption{\footnotesize{ Probability of missed detection versus probability of false alarm for the MRC detector  for  parameters $L=4$ and $J=4$, and for different values of $A, N$, and $M$. All curves were obtained via simulation. }}
\label{Fig:MN}
\vspace{-0.5cm}
\end{figure}

%%Fig8
%%%%***
%\begin{figure}
%\centering 
%\includegraphics[width=0.62\textwidth] {Fig/compareJustDecreasing_A_variable_n5_m4_J4_Pfa.eps}
%\vspace{-0.6cm}
%\caption{ $\Pm$ versus $A$, for parameters $L=4, J=4, M=4$, $N=5$ and $\Pf=0.1$ and different molecules. }
%\label{fig.8}
%\vspace{-0.1cm}
%\end{figure}
%
%
%In order to better investigate the effect of SNR, Fig. \ref{fig.8} depicts the probability of missed detection versus parameter $A$ for a given probability of false alarm of $\Pf=0.1$, with the model parameters of $L=4, J=4, M=4$, and $N=5$.
%It can be observed that the performance of Max-Log detector is almost near the optimum detector, indicating that the Max-Log detector is approximately optimal for large range of SNRs. In addition, by increasing the SNR, after a specific value (here $A \approx 17$), the probabilities of missed detection for none of the detectors further decrease. That is, increasing the SNR does not make any considerable improvement in the performance of the detectors. This event was previously observed in Fig. \ref{fig.4} where increasing the SNR does not significantly improve the performance of detectors. 

%%%%%%%%% Asymptotic
%Fig9&10
In the following,  we validate the asymptotic analysis given in Section~V-C. Again, we provide results only for the MRC detector for clarity of presentation. Moreover, we choose the detection threshold $\gamma$ such that $\Pf$ and $\Pm$ are identical. 
In Fig.~\ref{Fig:Exp}, we show the error exponent $\mathsf{Ex}=\mathsf{Ex}_0=\mathsf{Ex}_1$ versus $|s|$ for  $p_0 = p_1 = 0.1$, $N=1$, and different values of $A\in\{4,6,8,10\}$. It can be observed from Fig.~\ref{Fig:Exp} that the value of $s^*$ depends on the value of $A$. In addition, as $A$ increases, the optimal $s^*$ decreases and the exponent increases. To evaluate the tightness of the upper bounds in (\ref{Bound222}), in Fig.~\ref{Fig:ExpComp}, we show the probability of error, $\Pf=\Pm$, obtained via simulation and the upper bound for optimal $s^*$ obtained from (\ref{s_opt})  versus $M$ for $A=4$, $J=4$, $N=1$, and $p_0 = p_1 = 0.1$. We observe from Fig.~\ref{Fig:ExpComp} that the derived upper bound in (\ref{Bound222}) becomes tight when $M$ is large. That is, as $M$ increases, the difference in the actual error probability and the derived upper bound decreases which verifies the derivations in Section~V-C. Moreover, as $A$ increases, the slope of the error rate curve increases. This is due to the fact that as $M$ increases, the reporting channel becomes the bottleneck for  the overall detection performance and increasing $A$ improves the reliability of  the reporting channel. 
%
%%%***
\begin{figure} 
\centering 
\includegraphics[width=0.6\textwidth]{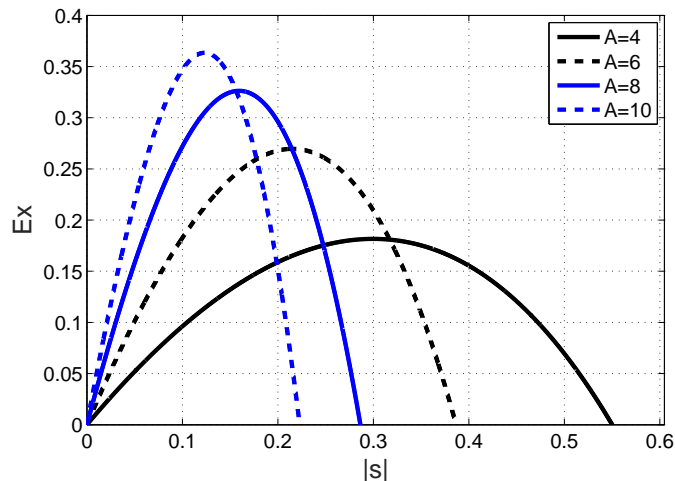} %Ex_Analysis  %
\vspace{-0.8cm}
\caption{Error exponent versus $|s|$ for $J=4, p_0 = p_1 = 0.1$, and different $A\in\{4,6,8,10\}$.  }
\label{Fig:Exp} 
\vspace{-0.4cm}
\end{figure}
%%%***
\begin{figure}
\centering 
\includegraphics[width=0.6\textwidth] {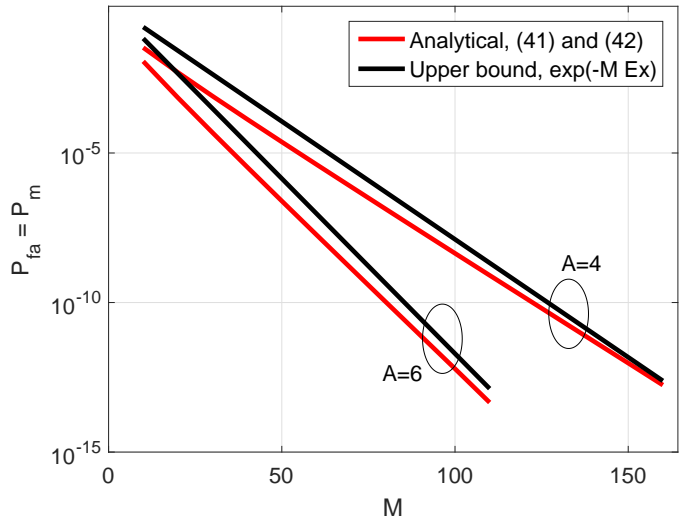} %  %Ex_MRC
\vspace{-0.7cm}
\caption{ Probability of error, $\Pf=\Pm$, versus $M$ for MRC detector, $J=4$, $p_0=p_1=0.1$, and $A\in\{4,6\}$. }
\label{Fig:ExpComp}
\vspace{-0.7cm}
\end{figure}

\section{Conclusions}

In this paper, we studied abnormality detection via diffusive MCs  where several sensors send their sensed values to an FC via MC links. We considered the DTM and STM reporting schemes where the sensors employ different types of molecules and the same type of molecules, respectively. The final decision about whether or not an abnormality has occurred is made at the FC based on the number of molecules received from the sensors. For this collaborative detection system, we first derived the optimal decision rules for both DTM and STM. For STM, we showed that the optimal detector has a simple equivalent form which is easier to implement. Unfortunately, for DTM, the optimal detector is complex and cannot be transformed into  a simple equivalent form. Therefore, we also developed several sub-optimal detectors for DTM. Furthermore, we analyzed the performance of the proposed detectors in terms of their false alarm and missed detection probabilities. Simulation results verified our analytical derivations and provided interesting insights for system design. For example, although STM is generally simpler than DTM, it outperforms DTM when a large number of noise molecules are present in the environment. Moreover, for the proposed sub-optimal DTM detectors,  the Max-Log detector offers the best performance, and the relative performance of the MRC, CV, and two-stage detectors crucially depends  on the choice of system parameters.

%%******------------------------------------
%%******------------------------------------
\appendix
%%******------------------------------------
\subsection{Proof of Theorem \ref{Th2}} \label{Th2:Proof} %%\ref{newapp}
We have to prove the monotonicity of the $\mathsf{LLR}(\vec{y}^{\{m\}})$ in (\ref{llr_loc}) in terms of $\sigma_{y}^{\{m\}}=\sum_{n=1}^{N}y^{\{m\}}_n$. Without loss of generality, we prove the monotonicity of the likelihood ratio instead of the LLR. Albeit being a discrete function, our proof assumes that $\sigma_{y}^{\{m\}}$ is continuous; however, if a function is monotonic with respect to a continuous argument, the monotonicity also applies to discrete arguments. First, we write the likelihood ratio of each sensor according to (\ref{llr_loc}) as
\begin{align}\label{temp_lr}
\text{e}^{\mathsf{LLR}\left(\sigma_{y}^{\{m\}}\right)}=\frac{\sum_{x_m \in{\mathcal{X}}}\g_1\left(x_m\right)\Ex{\left(-N\left(x_mA+J\right)\right)}\left(x_mA+J\right)^{\sigma_{y}^{\{m\}}}}{\sum_{x_m\in{\mathcal{X}}}\g_0\left(x_m\right)\Ex{\left(-N\left(x_mA+J\right)\right)}\left(x_mA+J\right)^{\sigma_{y}^{\{m\}}}}.
\end{align}
After taking the derivative of (\ref{temp_lr}) with respect to $\sigma_{y}^{\{m\}}$, since the denominator of the derivative is positive, we only need to investigate whether the numerator of the derivative is also positive. The numerator of the derivative can be written as
\begin{align}
\text{Num}=\sum_{x_m\in{\mathcal{X}}}\sum_{x'_m\in{\mathcal{X}}}&\Ex{\left(-N\left(\left(x_m+x'_m\right)A+J\right)\right)}\left(x_mA+J\right)^{\sigma_{y}^{\{m\}}}\left(x'_mA+J\right)^{\sigma_{y}^{\{m\}}} \nonumber \\
&\times \g_1\left(x_m\right)\g_0\left(x'_m\right)\mathsf{log}{\left(\frac{x_mA+J}{x'_mA+J}\right)}.
\end{align}
Obviously, except for the $\mathsf{log}$ term, all other terms are positive. Hence, we separate the inner sum into terms with $x'_m>x_m$ and $x'_m<x_m$, respectively. Accordingly, one can define
\begin{align}
\text{LHS}=\sum_{x_m\in{\mathcal{X}}}\sum_{x'_m\in{\mathcal{X}}, x_m>x'_m}&\Ex{\left(-N\left(\left(x_m+x'_m\right)A+J\right)\right)}\left(x_mA+J\right)^{\sigma_{y}^{\{m\}}}\left(x'_mA+J\right)^{\sigma_{y}^{\{m\}}} \nonumber\\
&\times \g_1\left(x_m\right)\g_0\left(x'_m\right)\mathsf{log}{\left(\frac{x_mA+J}{x'_mA+J}\right)}
\end{align}
and
\begin{align}\label{nested}
\text{RHS}=\sum_{x_m\in{\M}}\sum_{x'_m\in{\mathcal{X}}, x'_m>x_m}&\Ex{\left(-N\left(\left(x_m+x'_m\right)A+J\right)\right)}\left(x_mA+J\right)^{\sigma_{y}^{\{m\}}}\left(x'_mA+J\right)^{\sigma_{y}^{\{m\}}}\nonumber \\ 
&\times \g_1\left(x_m\right)\g_0\left(x'_m\right)\mathsf{log}{\left(\frac{x'_mA+J}{x_mA+J}\right)}.  
\end{align}
Then, Num$=$LHS-RHS, which is positive only if LHS$\geq$RHS. By exchanging the two summations and variables $x_m$ and $x'_m$ in (\ref{nested}), we obtain
\begin{multline}
\text{LHS}-\text{RHS}=\sum_{x_m\in{\M}}\sum_{x'_m\in{\M}, x_m>x'_m}\Ex{\left(-N\left(\left(x_m+x'_m\right)A+J\right)\right)}\left(x_mA+J\right)^{\sigma_{y}^{\{m\}}}\left(x'_mA+J\right)^{\sigma_{y}^{\{m\}}}\\
\times \mathsf{log}{\left(\frac{x_mA+J}{x'_mA+J}\right)}\left[\g_1\left(x_m\right)\g_0\left(x'_m\right)-\g_1\left(x'_m\right)\g_0\left(x_m\right)\right], \label{LR}
\end{multline} 
where all terms except the one in brackets are positive. Therefore, a sufficient condition for LHS$-$RHS$\geq0$ is
\begin{align}
\frac{\g_1\left(x_m\right)}{\g_1\left(x'_m\right)}\geq \frac{\g_0\left(x_m\right)}{\g_0\left(x'_m\right)},\quad \forall x_m, x'_m\in{\M}, \forall x_m>x'_m.
\end{align}
This completes the proof.

%%******------------------------------------
\subsection{Proof of Theorem \ref{theorem_general_Pfa_d} } \label{Th1:Proof}
%\begin{proof}
By defining $\Phi^{\{m\}}\left(s\vert \HTemp_i\right), i=0, 1,$ as the Laplace transform of the PDF of the LLR of sensor~$m$, we have
\begin{align*}
\Phi^{\{m\}}\left(s\vert \HTemp_i \right)=\E \left(\Ex \left(s\mathsf{LLR}\left(\sigma_{y}^{\{m\}}\right)\right)\vert \HTemp_i \right)=\sum_{x_m \in{\M}}\E \left(\Ex \left(s\mathsf{LLR} \left(\sigma_{y}^{\{m\}}\right)\right) \vert \HTemp_i, x_m \right)\g_i\left(x_m \right),
\end{align*}
where
\begin{align}\label{eq}
\E\left(\Ex\left(s\mathsf{LLR}\left(\sigma_{y}^{\{m\}}\right)\right) \vert \HTemp_i, x_m\right)=\sum_{w}\Ex\left(s\mathsf{LLR}\left(w\right)\right)\text{P}\left(\sigma_{y}^{\{m\}}=w \vert \HTemp_i, x_m\right),
\end{align}
and the parameters in (\ref{eq}) are defined in Theorem \ref{theorem_general_Pfa_d}. Conditioned on $x_m$, the term $\text{P}\big(\sigma_{y}^{\{m\}}=w \vert \HTemp_i, x_m\big)$ looses its dependence on $\HTemp_i$. After some calculations, we have
\begin{align}
\Phi^{\{m\}}\left(s\vert \HTemp_i\right)=\E\left(\Ex\left(s\mathsf{LLR}\left(\sigma_{y}^{\{m\}}\right)\right)\vert \HTemp_i\right)=\sum_{w}\Ex\left(s\mathsf{LLR}\left(w\right)\right)\mathsf{L}_i\left(w\right),
\end{align}
where $\mathsf{L}_i\left(w\right)$ is given in (\ref{T_h_difinition}). Now, since the molecules received from different sensors are independent, we obtain the Laplace transform of the PDF of the LLR for the $M$ sensors as
\begin{align}\label{eq64}
\Phi\left(s\vert \HTemp_i\right) = \sum_{\vec{\sigma}_y \in{\LLRset}}\mathsf{L}_i \left( \sigma_{y}^{\{1\}}\right) \cdots \mathsf{L}_i \left( \sigma_{y}^{\{m\}}\right)  \Ex \left(s\left( \mathsf{LLR} \left(\sigma_{y}^{\{1\}}\right)+\cdots
 +\mathsf{LLR}\left(\sigma_{y}^{\{M\}}\right)\right)\right),
\end{align}
where $\mathcal{S}\triangleq\Big\{ \vec{\sigma}_y = [\sigma_{y}^{\{1\}}, \cdots, \sigma_{y}^{\{M\}}]:~\sum_{m=1}^{M}\mathsf{LLR}\left(\sigma_{y}^{\{m\}}\right) > \gamma \Big\}$. The expression in (\ref{eq64}) is a weighted sum of exponential terms in $s$. Hence, by applying the inverse Laplace transform, each individual term in (\ref{eq64}) results in a shifted version of $\delta\left(\ell\right)$ with argument $\ell$. Hence, the PMF of the $\overline{\mathsf{LLR}}$ is given by
\begin{align}\label{eq65}
\text{P}_{\overline{\mathsf{LLR}}}(\ell\vert \mathcal{H}_i) = \sum_{\vec{\sigma}_y \in{\LLRset}}\mathsf{L}_i\left(\sigma_{y}^{\{1\}}\right)\cdots \mathsf{L}_i\left(\sigma_{y}^{\{M\}}\right)\delta\left(\ell+\left(\mathsf{LLR}\left(\sigma_{y}^{\{1\}}\right)\right)+\cdots +\mathsf{LLR}\left(\sigma_{y}^{\{M\}}\right)\right).
\end{align}
According to the general form of the probabilities of missed detection and false alarm given in (\ref{integral_in_TH2_1}), we have to sum $\text{P}_{\overline{\mathsf{LLR}}}(\ell\vert \mathcal{H}_i)$ over all possible $\ell$ from $\gamma$ to $\infty$ where only those terms in (\ref{eq65})   remain in the summation whose \emph{delta functions} are located in the following interval 
\begin{align}
\left[\mathsf{LLR}\left(\sigma_{y}^{\{1\}}\right)+\cdots+\mathsf{LLR}\left(\sigma_{y}^{\{M\}}\right)\right]>\gamma.
\end{align}
This leads to (\ref{Pfad_difinition}) in Theorem~2 and completes the proof.

\end{document}